\pdfoutput=1
\documentclass[11pt]{article}

\usepackage{a4}
\usepackage[top=30truemm,bottom=30truemm]{geometry}

\usepackage{authblk}
\usepackage{graphicx}
\usepackage{rotating}
\usepackage{multirow}

\usepackage{amsmath}
\usepackage{amssymb}
\usepackage{amsthm}
\usepackage{algorithm}
\usepackage[noend]{algpseudocode}
\usepackage{algorithmicx}
\usepackage{url}
\usepackage{enumerate}
\usepackage{bm}
\usepackage{setspace}
\usepackage{epsfig}
\usepackage{color}
\usepackage{subfigure}

\newcommand{\Suffix}{\mathit{Suffix}}

\newcommand{\STrie}{\mathit{STrie}}
\newcommand{\STree}{\mathit{STree}}
\newcommand{\DAWG}{\mathit{DAWG}}

\newcommand{\Endpos}{\mathit{Epos}}
\newcommand{\suflink}{\mathit{slink}}

\newcommand{\Long}{\mathit{long}}
\newcommand{\lspacer}{\kern.2ex}
\newcommand{\spacer}{\rule{0pt}{1.6ex}\kern.2ex}
\newcommand{\rev}[1]{\overline{#1}}

\newcommand{\longsuf}{\mathit{lrs}}
\newcommand{\LPT}{\mathit{LPT}}

\newcommand{\sw}{\mathsf{w}}
\newcommand{\wdeg}{\mathsf{w\_deg}}
\newcommand{\bottom}{\perp}

\newcommand{\SLT}{\mathit{SLT}}

\def\W#1#2{{\cal W}_{#1}({#2})}

\newtheorem{fact}{Fact}

\def\fn(#1){\fnbf{#1}}

\newtheorem{theorem}{Theorem}
\newtheorem{lemma}[theorem]{Lemma}
\newtheorem{corollary}[theorem]{Corollary}
\newenvironment{remark}{Remark}

\title{
  Fully-Online Suffix Tree and Directed Acyclic Word Graph\\ Construction for Multiple Texts
}

\author[1]{Takuya Takagi}
\author[2]{Shunsuke Inenaga}
\author[1]{Hiroki Arimura}
\author[ ]{Dany Breslauer}
\author[3]{Diptarama Hendrian}

\affil[1]{\textit{Graduate School of IST, Hokkaido University, Japan} \newline
\texttt{\small \{tkg, arim\}@ist.hokudai.ac.jp}
}
\affil[2]{\textit{Department of Informatics, Kyushu University, Japan} \newline
\texttt{\small inenaga@inf.kyushu-u.ac.jp}
}
\affil[3]{\textit{Graduate School of Information Sciences, Tohoku University, Japan} \newline
\texttt{\small diptarama@shino.ecei.tohoku.ac.jp}
}

\date{}

\begin{document}
\maketitle

\begin{abstract}
We consider construction of the {\em suffix tree} and 
the {\em directed acyclic word graph (DAWG)}
indexing data structures for a collection $\mathcal{T}$ of texts,
where a new symbol may be appended to any text in 
$\mathcal{T} = \{T_1, \ldots, T_K\}$, \emph{at any time}.
This {\em fully-online} scenario, which arises in dynamically indexing multi-sensor data,
is a natural generalization of the long solved
\emph{semi-online} text indexing problem, where 
texts $T_1, \ldots, T_{k}$ are permanently fixed
before the next text $T_{k+1}$ is processed for each $1 \leq k < K$.
We present fully-online algorithms that 
construct the \emph{suffix tree} and the \emph{DAWG} for $\mathcal{T}$ in 
$O(N \log \sigma)$ time and $O(N)$ space, where $N$ is the total lengths of the
strings in $\mathcal{T}$ and $\sigma$ is their alphabet size.
The standard explicit representation of 
the suffix tree leaf edges and some DAWG edges must be
relaxed in our fully-online scenario,
since too many updates on these edges are required in the worst case.
Instead, we provide access to the updated suffix tree leaf edge
labels
and the DAWG edges to be redirected via auxiliary data structures, 
in $O(\log \sigma)$ time per added character.\footnote{A preliminary conference version of this 
paper~\cite{TakagiIA16} contained an error in the analysis of the
fully-online DAWG algorithm, where the claimed $O(N \log \sigma)$ time bound 
neglected to account for the DAWG edge re-directions. 
We show that in some cases $\Theta(N\min(K, \sqrt N))$ or $\Theta(N^{1.5})$ DAWG edge re-directions may be structurally required,
and that the correct time bound for the algorithm in the
conference paper~\cite{TakagiIA16} is $O(N\min(K, \sqrt N)\log \sigma)$
(Lemma~\ref{lem:Weiner_lower_bound}).}\\

\end{abstract}


\section{Introduction}

Text indexing is a fundamental problem in computer science, which arises in many applications including text retrieval, molecular biology, signal processing, and sensor data analysis.
In this paper, we focus on indexing 
a collection of multiple texts,
so that subsequent pattern matching queries
can be answered quickly.
In particular, we study online indexing for a collection 
$\mathcal{T}$ of multiple texts,
where a new character can be appended to each text at \emph{any} time. 
Such fully-online indexing for multiple growing texts 
has potential applications to continuous processing of data streams, 
where a number of symbolic events or data items are produced from multiple, 
rapid, time-varying, and unbounded data streams~\cite{Babcock:Babu:Datar:Motwani:2002,KeoghLC02}. 
For example, motif mining system tries to discover characteristic or 
interesting collective behaviors, such as frequent path or anomalies, 
from data streams generated by a collection of moving objects or 
sensors~\cite{WangZX2014,KeoghLC02}.

In this paper we consider two fundamental text indexing data structures,
the \emph{suffix tree}~\cite{Weiner} and the \emph{directed acyclic word graph}
(\emph{DAWG})~\cite{blumer85:_small_autom_recog_subwor_text}.
The suffix tree of a string $T$ is an edge-labeled rooted tree
which represents all the suffixes of $T$ in space linear in the length of $T$,
while the DAWG of $T$ is the smallest (partial) DFA that
recognizes all suffixes of $T$ which also occupies linear space.
The suffix tree can be constructed in $O(n \log \sigma)$ time and $O(n)$ space,
in a right-to-left online manner by Weiner's algorithm~\cite{Weiner},
and in a left-to-right online manner by Ukkonen's algorithm~\cite{Ukkonen95},
where $n$ is the length of $T$ and $\sigma$ is the alphabet size.
The DAWG of a given text $T$ can also be built in
$O(n \log \sigma)$ time and $O(n)$ space,
in a left-to-right online manner by Blumer et al.'s algorithm~\cite{blumer85:_small_autom_recog_subwor_text}.
The ``duality'' of the DAWG of $T$ and the suffix tree of the reversal
$\rev{T}$ of $T$ is known,
more specifically, the tree of the \emph{suffix links} of the DAWG of $T$
coincides with the suffix tree of $\rev{T}$~\cite{blumer85:_small_autom_recog_subwor_text,cr:94}.
We note that this property also holds for multiple texts.

Let $N$ be the final total length of the growing texts
in a fully-online text collection $\mathcal{T} = \{T_1, \ldots, T_K\}$.
The above existing suffix tree and DAWG construction 
algorithms for a single text also work within the same $O(N\log \sigma)$ time and $O(N)$ space bounds for a collection of growing texts 
in the \emph{semi-online} setting, where only the last inserted text can be extended~\cite{gusfield:97,Blumer87}.
However, special attention is needed
for construction of the suffix tree and the DAWG
in our \emph{fully-online} setting.
For the fully-online \emph{right-to-left} setting
where new characters are prepended to the multiple texts,
we show that a matching upper and lower bound
$\Theta(N\min(K, \sqrt N))$ or $\Theta(N^{1.5})$ holds
for a direct extension of Weiner's original algorithm,
where $K$ is the number of texts in the collection.
This also implies that
up to $\Theta(N\min(K, \sqrt N))$ or $\Theta(N^{1.5})$ DAWG
edge re-directions can be required during the DAWG construction
in the fully-online \emph{left-to-right} setting.
Also, we show that during the construction of the suffix tree
for a fully-online left-to-right text collection,
the open-ended suffix tree leaf edge label representation, the cornerstone
of Ukkonen's~\cite{Ukkonen95} on-line suffix tree algorithm, may
have to update the association between the numerous suffix tree 
leaf edge labels and the various texts $Omega(\frac{N^2}{K})$ times,
which turns $\Omega(N^2)$ when the collection contains only a
constant number of texts.
Thus, if we wish to stay within the $O(N\log \sigma)$ time bounds in the \emph{fully-online} setting, 
the DAWG edges and the suffix tree leaf edge labels
in the fully-online left-to-right setting
cannot be directly explicitly maintained.
We call this as the \emph{leaf ownership} problem.

To overcome the above difficulties,
we first show how to extend Weiner's algorithm
in the fully-online right-to-left setting,
with the aid of the nearest marked ancestor (NMA) data structures~\cite{westbrook92:_fast_increm_planar_testin}.
The resulting algorithm runs in $O(N \log \sigma)$ time
and takes $O(N)$ total space for a general ordered alphabet of size $\sigma$.
We then show that how an $O(N)$-space representation of the DAWG
can be incrementally maintained for a left-to-right online text collection,
in overall $O(N \log \sigma)$ time and $O(N)$ space.
Hence, at any moment during the fully-online growth of the texts,
we can find all $occ$ occurrences of a given pattern of length $M$
in the current text collection in $O(M \log \sigma + occ)$ time,
using any of these two text indexing structures.

Our next goal is to extend Ukkonen's construction~\cite{Ukkonen95} 
to fully-online left-to-right construction of suffix trees for multiple texts
in $O(N \log \sigma)$ time and $O(N)$ space bounds.
As was already mentioned above, however,
it is not possible to explicitly maintain
the owners of the leaf edges in the suffix tree here.
To overcome this difficulty, we present a new novel technique
which \emph{swaps} the active points among the texts that involved in the
update of the suffix tree,
together with a query algorithm which efficiently answers
the owners of the particular leaf edges involved in the update.
As a result, we obtain a natural extension of Ukkonen's construction
where suffixes are inserted to the current tree
in decreasing order of their lengths (called the \emph{forward approach}).
We also present an alternative method that inserts
suffixes in increasing order of their lengths
each time a new character is appended to one of the texts
(called the \emph{backward approach}).
Both methods work in $O(N \log \sigma)$ time and $O(N)$ space,
with the aid of the extended Weiner algorithm
for right-to-left text collection.

\subsection*{Related work}
We note that we can obtain fully-online text indexing
data structure for multiple texts 
by using existing more general dynamic text indexing data structures as follows.
To use the indexing data structure of Ferragina and Grossi~\cite{FerraginaG99}
which permits character-wise updates, 
we build a text $\$_1 \cdots \$_K$ which initially consists only of $K$ delimiters.
Then, appending a character $a$ to the $k$th text in the collection 
reduces to prepending $a$ to the $k$th delimiter $\$_k$.
Using this approach, 
the text indexing data structure of Ferragina and Grossi~\cite{FerraginaG99}
takes $O(N \log N)$ total time to be constructed,
requires $O(N \log N)$ space,
and allows pattern matching in $O(M + \log N + N\log M + occ)$ time.
Using the compressed indexing data structure for
a dynamic text collection of Chan et al.~\cite{ChanHLS07}, 
we can append a new character $a$ to the $k$th text $T_k$ 
by removing $T_k$ and then adding $T_ka$ in $O(|T_k|)$ time. 
This yields a fully-online text indexing structure
with $O(N^2 \log N)$ construction time 
and $O(N)$ bits of space (or $O(N / \log N)$ words of space assuming
$\Theta(\log N)$-bit machine word), 
supporting pattern matching in $O(M\log N + occ\log^2 N)$ time.

\section{Preliminaries}

\subsection{String notations}

Let $\Sigma$ be a general ordered alphabet.
Any element of $\Sigma^*$ is called a \emph{string}.
For any string $T$, let $|T|$ denote its length.
Let $\varepsilon$ be the empty string, namely, $|\varepsilon| = 0$.
If $T = XYZ$, then $X$, $Y$, and $Z$ are called 
a \emph{prefix}, a \emph{substring}, and a \emph{suffix} of $T$, respectively.
For any $1 \leq i \leq j \leq |T|$,
let $T[i..j]$ denote the substring of $T$ that begins at position $i$
and ends at position $j$ in $T$.
For convenience, let $T[i..j] = \varepsilon$ if $i > j$.
For any $1 \leq i \leq |T|$, let $T[i]$ denote the $i$th character of $T$.
For any string $T$, let $\Suffix(T)$ denote the set of suffixes of $T$,
and for any set $\mathcal{T}$ of strings,
let $\Suffix(\mathcal{T})$ denote the set of suffixes of all strings in $\mathcal{T}$.
Namely, $\Suffix(\mathcal{T}) = \bigcup_{T \in \mathcal{T}} \Suffix(T)$.
For any string $T$, let $\rev{T}$ denote the reversed string of $T$,
i.e., $\rev{T} = T[|T|] \cdots T[1]$.

Let $\mathcal{T} = \{T_1, \ldots, T_K\}$ be a collection of $K$ texts.
For each text $T_k$, the integer $k$ is called its \emph{id}.
For any $1 \leq k \leq K$,
let $\longsuf_{\mathcal{T}}(T_k)$ be the longest repeating suffix of $T_k$
that occurs at least twice in the texts of $\mathcal{T}$.

\subsection{Suffix trees and DAWGs for multiple texts}
\label{sec:def_ST_DAWG}

\subsubsection{Suffix tries}

The \emph{suffix trie} for a text collection $\mathcal{T} = \{T_1, \ldots, T_K\}$,
denoted $\STrie(\mathcal{T})$, is a trie which represents 
$\Suffix(\mathcal{T})$.
The size of $\STrie(\mathcal{T})$ is $O(N^2)$,
where $N$ is the total length of texts in $\mathcal{T}$.
We identify each node $v$ of $\STrie(\mathcal{T})$
with the string that $v$ represents.
A substring $x$ of a text in $\mathcal{T}$ is said to be \emph{branching} in $\mathcal{T}$,
if there exist two distinct characters $a, b \in \Sigma$ such that
both $xa$ and $xb$ are substrings of some texts in $\mathcal{T}$.
Clearly, node $x$ of $\STrie(\mathcal{T})$
is branching iff $x$ is branching in $\mathcal{T}$.

For each node $av$ of $\STrie(\mathcal{T})$ with $a \in \Sigma$ and $v \in \Sigma^*$,
let $\suflink(av) = v$.
This auxiliary edge $\suflink(av) = v$ from $av$ to $v$ is called a \emph{suffix link}.
We define the \emph{reversed suffix link}
$\W{a}{v} = av$ iff $\suflink(av) = v$.
For any node $v$ and $a \in \Sigma$,
if $av$ is not a substring of the texts in $\mathcal{T}$,
then $\W{a}{v}$ is undefined.
By definition, the reversed suffix links on $\STrie(\mathcal{T})$
form a rooted tree which coincides with $\STrie(\rev{\mathcal{T}})$,
the suffix trie for the collection $\rev{\mathcal{T}} = \{\rev{T_1}, \ldots, \rev{T_K}\}$ of the reversed texts.

\begin{figure}[t]
  \centerline{
    \includegraphics[width=1.0\linewidth]{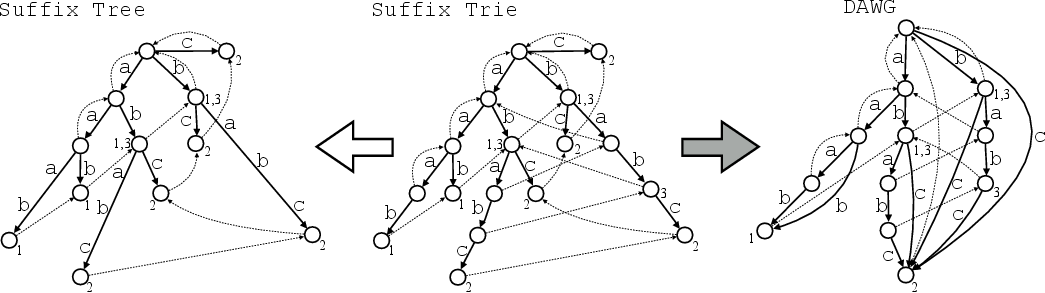}
  }
  \caption{Illustration for $\STrie(\mathcal{T})$, $\STree(\mathcal{T})$,
  and $\DAWG(\mathcal{T})$ with
  $\mathcal{T} = \{T_1 = \mathtt{aaab}, T_2 = \mathtt{ababc}, T_3 = \mathtt{bab}\}$.
  The solid arrows and broken arrows represent 
  the edges and the suffix links of each data structure, respectively.
  The number $k$~($k = 1, 2, 3$) beside each node indicates that 
  the node represents a suffix of $T_k$.
  The nodes $[\mathtt{ab}]_{\mathcal{T}}$ and $[\mathtt{b}]_{\mathcal{T}}$ are separated
  in $\DAWG(\mathcal{T})$ since the node $\mathtt{bab}$ in $\STrie(\mathcal{T})$ is represents
  a suffix of $T_3$, while the node $\mathtt{abab}$ does not 
  (see also the subtrees rooted at nodes $\mathtt{ab}$ and $\mathtt{b}$ in $\STrie(\mathcal{T})$).
  }
  \label{fig:trie_tree_dawg}
\end{figure}

\subsubsection{Suffix trees}

The \emph{suffix tree}~\cite{Weiner} 
for a text collection $\mathcal{T}$, denoted $\STree(\mathcal{T})$,
is a ``compacted trie'' which represents $\Suffix(\mathcal{T})$.
$\STree(\mathcal{T})$ is 
obtained by compacting every path of $\STrie(\mathcal{T})$
which consists of non-branching internal nodes (see Fig.~\ref{fig:trie_tree_dawg}).
Since every internal node of $\STree(\mathcal{T})$ is branching,
and since there are at most $N$ leaves in $\STree(\mathcal{T})$,
the numbers of edges and nodes are $O(N)$.
The edge labels of $\STree(\mathcal{T})$ are
non-empty substrings of some text in $\mathcal{T}$.
By representing each edge label $x$ with a triple $\langle k, i, j \rangle$ of integers
s.t. $x = T_k[i..j]$, $\STree(\mathcal{T})$ can be stored with $O(N)$ space.
We say that any branching (resp. non-branching) substring of $\mathcal{T}$
is an \emph{explicit node} (resp. \emph{implicit node}) of $\STree(\mathcal{T})$.
An implicit node $x$ is represented by a triple $(v, a, \ell)$, 
called a \emph{reference} to $x$,
such that $v$ is an explicit ancestor of $x$, 
$a$ is the first character of the path from $v$ to $x$, 
and $\ell$ is the length of the path from $v$ to $x$. 
A reference $(v, a, \ell)$ to node $x$ is called 
\emph{canonical} if $v$ is the lowest explicit ancestor of $x$.

For each explicit node $av$ of $\STree(\mathcal{T})$ with
$a \in \Sigma$ and $v \in \Sigma^*$, let $\suflink(av) = v$.
For each explicit node $v$ and $a \in \Sigma$,
we also define the reversed suffix link $\W{a}{v} = avx$
where $x \in \Sigma^*$ is the shortest string
such that $avx$ is an explicit node of $\STree(\mathcal{T})$.
$\W{a}{v}$ is undefined if $av$ is not a substring of texts in $\mathcal{T}$.
These reversed suffix links are also called as \emph{Weiner links}
(or \emph{W-link} in short) in the literature~\cite{BreslauerI13}.
A W-link $\W{a}{v} = avx$ is said to be \emph{hard}
if $x = \varepsilon$, and \emph{soft} if $x \in \Sigma^+$.
Let $\sw$ be a Boolean function such that 
for any explicit node $v$ and $a \in \Sigma$,
$\sw_a(v) = 1$ iff (soft or hard) W-link $\W{a}{v}$ exists.
Notice that if $\sw_a(v) = 1$ for a node $v$ and $a \in \Sigma$,
then $\sw_a(u) = 1$ for every ancestor of $v$.

\subsubsection{Directed acyclic word graphs (DAWGs)}

The \emph{directed acyclic word graph} (\emph{DAWG} in short)~\cite{blumer85:_small_autom_recog_subwor_text,Blumer87}
of a text collection $\mathcal{T}$, denoted $\DAWG(\mathcal{T})$,
is a smallest DAG which represents $\Suffix(\mathcal{T})$.
$\DAWG(\mathcal{T})$ is 
obtained by merging identical subtrees of $\STrie(\mathcal{T})$
connected by the suffix links (see Fig.~\ref{fig:trie_tree_dawg}).
Hence, 
the label of every edge of $\DAWG(\mathcal{T})$ is a single character.
The numbers of nodes and edges of $\DAWG(\mathcal{T})$ are $O(N)$~\cite{Blumer87},
and hence $\DAWG(\mathcal{T})$ can be stored with $O(N)$ space.
$\DAWG(\mathcal{T})$ can be defined formally as follows:
For any string $x$, let $\Endpos_{\mathcal{T}}(x)$
be the set of ending positions of $x$ in the texts in $\mathcal{T}$,
i.e.,
\[
\Endpos_{\mathcal{T}}(x) = \{(k, j) \mid x = T_k[j-|x|+1..j], 1 \leq j \leq |T_k|, 1 \leq k \leq K\}.
\]
Consider an equivalence relation $\equiv_{\mathcal{T}}$ on substrings $x,y$ of texts in $\mathcal{T}$
such that 
$x \equiv_{\mathcal{T}} y$ iff $\Endpos_{\mathcal{T}}(x) = \Endpos_{\mathcal{T}}(y)$.
For any substring $x$ of texts of $\mathcal{T}$, 
let $[x]_{\mathcal{T}}$ denote the equivalence class w.r.t. $\equiv_{\mathcal{T}}$.
There is a one-to-one correspondence between 
each node $v$ of $\DAWG(\mathcal{T})$ 
and each equivalence class $[x]_{\mathcal{T}}$,
and hence we will identify each node $v$ of $\DAWG(\mathcal{T})$
with its corresponding equivalence class $[x]_{\mathcal{T}}$.
Let $\Long([x]_{\mathcal{T}})$ denote the longest member of $[x]_{\mathcal{T}}$.
By the definition of equivalence classes,
$\Long([x]_{\mathcal{T}})$ is unique for each $[x]_{\mathcal{T}}$
and every member of $[x]_{\mathcal{T}}$ is a suffix of $\Long([x]_{\mathcal{T}})$.
If $x, xa$ are substrings of some text in $\mathcal{T}$ with 
$x \in \Sigma^*$ and $a \in \Sigma$,
then there exists an edge labeled with character $a \in \Sigma$
from node $[x]_{\mathcal{T}}$ to node $[xa]_{\mathcal{T}}$.
This edge is called \emph{primary} if $|\Long([x]_{\mathcal{T}})| + 1 = |\Long([xa]_{\mathcal{T}})|$,
and is called \emph{secondary} otherwise.
For each node $[x]_{\mathcal{T}}$ of $\DAWG(\mathcal{T})$ with $|x| \geq 1$,
let $\suflink([x]_{\mathcal{T}}) = y$, where
$y$ is the longest suffix of $\Long([x]_{\mathcal{T}})$ which does not belong to 
$[x]_{\mathcal{T}}$.
In the example of Fig.~\ref{fig:trie_tree_dawg},
$[\mathtt{aaab}]_{\mathcal{T}} = \{\mathtt{aaab}, \mathtt{aab}\}$.
The edge labeled with $\mathtt{b}$ from node $[\mathtt{aaa}]_{\mathcal{T}}$ to node $[\mathtt{aaab}]_{\mathcal{T}}$
is primary, while the edge labeled with $\mathtt{b}$ 
from $[\mathtt{aa}]_{\mathcal{T}}$
to node $[\mathtt{aaab}]_{\mathcal{T}}$ is secondary.
$\suflink([\mathtt{aaab}]_{\mathcal{T}}) = [\mathtt{ab}]_{\mathcal{T}}$.

\subsubsection{Duality of suffix trees and DAWGs}

There exists a nice duality between suffix trees and DAWGs.
To observe this,
it is convenient to consider the collection
$\rev{\mathcal{T}}$ of the reversed texts
each of which begins with a special marker $\$_i$,
i.e., $\rev{\mathcal{T}} = \{\$_1\rev{T_1}, \ldots, \$_K \rev{T_K}\}$.
For ease of notation, let $S_k = \rev{\mathcal{T}_k}$ for $1 \leq k \leq K$
and $\mathcal{S} = \{\$_1 S_1, \ldots, \$_K S_K\} = \rev{\mathcal{T}}$.
Then, it is known (c.f.~\cite{blumer85:_small_autom_recog_subwor_text,Blumer87,cr:94})
that the reversed suffix links of
$\DAWG(\mathcal{S})$ coincide with the suffix tree
$\STree(\mathcal{T})$ for the original text collection $\mathcal{T}$.
This fact can also be observed from the other direction.
Namely, the hard (resp. soft) W-links of $\STree(\mathcal{T})$ coincide with
the primary (resp. secondary) edges of $\DAWG(\mathcal{S})$.

Intuitively, this duality holds because
\begin{enumerate}
  \item[(1)] The reversed suffix links of $\STrie(\mathcal{S})$ form $\STrie(\mathcal{T})$ (and vice versa), and
  \item[(2)] When we construct $\DAWG(\mathcal{S})$ from $\STrie(\mathcal{S})$, we merge isomorphic subtrees that are connected by suffix links.
    During this merging process, the reversed suffix links get compacted and
    the resulting compacted links form the edges of $\STree(\mathcal{T})$.
\end{enumerate}

Using this duality, we can immediately show that the total number of
hard and soft W-links is linear in the total text length $N$,
since the number of edges of the DAWG is linear in $N$.
This also means that we can easily maintain the Boolean indicator $\sw$
with $O(N)$ space, so that $\sw_{a}(v)$ for a given node $v$ and $a \in \Sigma$
can be answered in $O(\log \sigma)$ time
(e.g., at each node $v$ we can maintain a BST storing only the characters $c$
s.t. $\sw_{c}(v) = 1$.)


\subsection{Fully-online text collection}
We consider a collection $\{T_1, \ldots, T_K\}$ of $K$ growing texts,
where each text $T_k$~($1 \leq k \leq K$) is initially the empty string $\varepsilon$.
Given a pair $(k, a)$ of a text id $k$ and a character $a \in \Sigma$
which we call an \emph{update operator},
the character $a$ is appended to the $k$-th text of the collection.
For a sequence $U$ of update operators,
let $U[1..i]$ denote the sequence of the first $i$ update operators in $U$
with $0 \leq i \leq |U|$.
Also, for $0 \leq i \leq |U|$
let $\mathcal{T}_{U[1..i]}$ denote the collection of texts
which have been updated according to the first $i$ update operators of $U$.
For instance, consider a text collection of three texts 
which grow according to the following sequence
$U = (1, \mathtt{a}), (2, \mathtt{b}), (2, \mathtt{a}), (3, \mathtt{a}), 
(1, \mathtt{a}), (3, \mathtt{c}), (3, \mathtt{b}), (2, \mathtt{b}), (1, \mathtt{a}),  
(1, \mathtt{b}), \\ (3, \mathtt{c}), (3, \mathtt{b}), (1, \mathtt{c}),  
(3, \mathtt{b}), (2, \mathtt{c})$ of 15 update operators.
Then, 
\[
{\arraycolsep = 1mm
 \mathcal{T}_{U[1..0]} = \left\{
  \begin{array}{c}
     \overset{\spacer}{\varepsilon} \\
     \overset{\spacer}{\varepsilon} \\
     \overset{\spacer}{\varepsilon} 
  \end{array}
  \right\},\ \ldots, \
 \mathcal{T}_{U[1..14]} = \left\{
  \begin{array}{cccccc}
     \overset{1}{\lspacer \mathtt{a}\spacer} & \overset{5}{\lspacer \mathtt{a}\spacer} & \overset{9}{\lspacer \mathtt{a}\spacer} & \overset{10}{\lspacer \mathtt{b}\spacer} & \overset{13}{\lspacer \mathtt{c}\spacer} & \\
     \overset{2}{\lspacer \mathtt{b}\spacer} & \overset{3}{\lspacer \mathtt{a}\spacer} & \overset{8}{\lspacer \mathtt{b}\spacer} & & & \\
     \overset{4}{\lspacer \mathtt{a}\spacer} & \overset{6}{\lspacer \mathtt{c}\spacer} & \overset{7}{\lspacer \mathtt{b}\spacer} & \overset{11}{\lspacer \mathtt{c}\spacer} & \overset{12}{\lspacer \mathtt{b}\spacer} & \overset{14}{\lspacer \mathtt{b}\spacer}
  \end{array}
  \right\}, \
 \mathcal{T}_{U[1..15]} = \left\{
  \begin{array}{cccccc}
     \overset{1}{\lspacer \mathtt{a}\spacer} & \overset{5}{\lspacer \mathtt{a}\spacer} & \overset{9}{\lspacer \mathtt{a}\spacer} & \overset{10}{\lspacer \mathtt{b}\spacer} & \overset{13}{\lspacer \mathtt{c}\spacer} & \\
     \overset{2}{\lspacer \mathtt{b}\spacer} & \overset{3}{\lspacer \mathtt{a}\spacer} & \overset{8}{\lspacer \mathtt{b}\spacer} & \overset{15}{\lspacer \mathtt{c}\spacer} & & \\
     \overset{4}{\lspacer \mathtt{a}\spacer} & \overset{6}{\lspacer \mathtt{c}\spacer} & \overset{7}{\lspacer \mathtt{b}\spacer} & \overset{11}{\lspacer \mathtt{c}\spacer} & \overset{12}{\lspacer \mathtt{b}\spacer} & \overset{14}{\lspacer \mathtt{b}\spacer}
  \end{array}
  \right\} 
}\]
where the superscript $i$ over each character $a$ in the $k$-th text 
implies that $U[i] = (k, a)$.
For instance, $U[15] = (2, \mathtt{c})$ and hence $\mathtt{c}$
was appended to the 2nd text $T_2 = \mathtt{bab}$ in $\mathcal{T}_{U[1..14]}$, 
yielding $T_2 = \mathtt{babc}$ in $\mathcal{T}_{U[1..15]}$. 

If there is no restriction on $U$ like the one in the example above,
then $U$ is called \emph{fully-online}.
If there is a restriction on $U$ such that once a new character is appended to the $k$-th text,
then no characters will be appended to its previous $k-1$ texts,
then $U$ is called \emph{semi-online}.
Hence, any semi-online sequence of update operators is 
of form
$$(1, T_1[1]), \ldots, (1, T_1[|T_1|]), \ldots, (K, T_K[1]), \ldots, (K, T_K[|T_K|]).$$

When we talk about the duality of suffix trees and DAWGs
in our fully-online scenario,
$\mathcal{S}_{U[1..i]}$ represents the set of the reversed texts from
$\mathcal{T}_{U[1..i]}$.


\section{Fully-online version of DAWG and Weiner's suffix tree algorithm}
\label{sec:weiner}

Blumer et.~al.~\cite{blumer85:_small_autom_recog_subwor_text,Blumer87} and
Crochemore~\cite{crochemore:86} introduced the DAWG, also called suffix automaton, and 
gave a DAWG construction algorithm for a collection of semi-online texts. 
Their DAWG construction algorithm is very closely related to Weiner's reverse right-to-left 
suffix tree construction algorithm~\cite{cr:94,gusfield:97,lothaire-book:2005,Weiner}.
In fact, both algorithms build dual structures and each exposes different parts of these structures,
where the collection of semi-online left-to-right text inputs to the DAWG algorithm
can be perceived as the same texts reversed right-to-left inputs to Weiner's suffix tree algorithm.
Blumer et al.'s algorithm does not require a terminating $\$$ symbol and it was noted
that the set of nodes of the DAWG and the reverse string's suffix tree coincide 
if the terminator symbols are present in both sets of inputs.

\subsection{Semi-online construction of Weiner's suffix trees and DAWGs}

We briefly explain how 
the suffix tree of a collection of semi-online right-to-left texts
can be built by using Weiner's algorithm.
For convenience, we assume that
there is an auxiliary node $\bottom$ that is the parent of the root $r$.
We also assume that the edge from $\bottom$ to $r$
is labeled with \emph{any} character $c$ from $\Sigma$,
$\W{c}{\bottom} = r$, and $\suflink(r) = \bottom$.
Assume that we have constructed
$\STree(\{T_1 \$_1, \ldots, T_{k-1} \$_{k-1}\})$
in which all the hard W-links have been constructed
and the Boolean indicator $\sw$ have been appropriately maintained.
Now we process the $k$-th and extend it from right-to-left.
Since the end-marker $\$_k$ is a unique character,
a new leaf representing $\$_k$ is created.
Suppose we have inserted the leaves for the suffixes of
$T_k \$_k$ with $T_k \in \Sigma^*$.
The leaf that represents the $k$-th text $T_k \$_k$
is called the \emph{handle leaf} for $T_k \$_k$. 
Now we are to prepend a new character $a$
and insert the extended text $aT_k\$_k$ to the tree.
We begin with the handle leaf $\ell$ for $T_k\$_k$.
We walk up from the handle leaf $\ell$
until finding the lowest explicit ancestor $u'$ of $\ell$
which has hard W-link $\W{a}{u'}$ defined for the added character $a$.
Also, let $u$ be the lowest explicit ancestor of $\ell$
such that $\sw_{a}(u) = 1$.
Note that $u$ is a descendant of $u'$.
Let $b$ be the first character of the path label from $u'$ to $\ell$.
We move to the node $v' = au'$ using the hard W-link $\W{a}{u'}$,
and let $v'' = au'by$ be the child of $v'$ below the edge
whose label begins with $b$, where $y \in \Sigma^*$.
There are two cases:
(1) If $|v''|-|v'| = |au'by| - |au'| = |by| > |u|-|u'|$,
then we create a new explicit node $v = v''[1..|u|+1]$
and set $\W{a}{u} = v$.
(2) Otherwise ($|by| = |u|-|u'|$),
then there already exists an explicit node $v''[1..|u|+1]$
and let $v$ be this node.
In both cases, we insert a new leaf $\ell'$
representing $aT_k\$_k$ as a child of $v$,
and create a new hard W-link $\W{a}{\ell} = \ell'$.
This insertion point $v$ for $\ell'$
represents the longest prefix of $aT_k\$_k$
that appears at least twice in the updated text collection,
and hence, $v$ is sometimes called as the \emph{longest repeating prefix} of $aT_k\$_k$.
Let $s$ be any node in the path from $u$ to $\ell$ such that $s \neq u$
(if any).
In the suffix tree before the text $T_k \$_k$ was extended with $a$,
we had $\sw_{a}(s) = 0$.
Now in the updated suffix tree, 
we update $\sw_{a}(s) = 1$ due to the insertion of the new handle leaf $\ell'$
which represents $a T_k \$_k$.
Also, node $s$ gets a new soft W-link $\W{a}{s} = \ell'$.
\hspace{-10truemm}These updates are common to both of Cases (1) and (2).
There can be further updates in Case (1):
Let $s'$ be any node in the path from $u'$ to $u$
such that $s' \neq u'$ and $s' \neq u$ (if any).
In the suffix tree before the text $T_k \$_k$ was extended with $a$,
node $s'$ had a soft W-link $\W{a}{s'} = v''$.
Now in the updated suffix tree,
this soft W-link is redirected as $\W{a}{s'} = v$.
Also, the soft W-link $\W{a}{u} = v''$ in the previous suffix tree
gets redirected and becomes the hard W-link $\W{a}{u} = v$
in the updated suffix tree.
See Figure~\ref{fig:weiner} for illustration.

\begin{figure}[htb]
  \hspace{20pt}
  \subfigure[]{%
    \includegraphics[height=4cm]{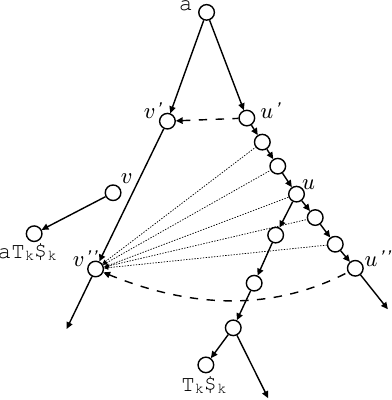}}%
  \hspace{20pt}
  \subfigure[]{%
    \includegraphics[height=4cm]{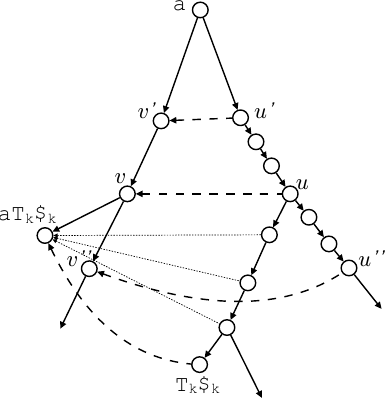}}%
  \hspace{20pt}
  \subfigure[]{%
    \includegraphics[height=4cm]{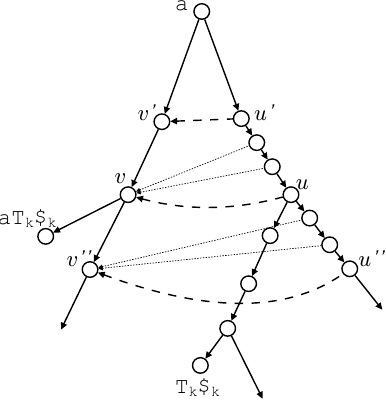}}%
  \vspace*{-0.1cm}
\caption{\small{Extending the text $T_k\$_k$ to $a T_k\$_k$.
Soft W-links are shown short-dashed and hard W-links are shown long-dashed.
(a) relevant existing W-links before extending. 
(b) New W-links pointing to $aT_k\$_k$ are created from all nodes on the 
path between $T_k\$_k$ up to $u$. 
(c) Existing W-links pointing to $v''$ from all nodes
on the path between $u$ up to $u'$ are 
redirected to point to $v$ instead of $v''$.
The new hard W-link $\W{a}{T_k\$_k}=aT_k\$_k$ and redirected hard W-link $\W{a}{u}=v$
have corresponding nodes on the path to $aT_k\$_k$, while all the other new soft W-links
involved point to $aT_k\$_k$ and the redirected soft W-links involved point to $v$.
The new node $v$ also adopts all the outgoing W-links from $v''$ (not shown).
\label{fig:weiner}}}
\end{figure}

Weiner's original algorithm is designed for a single right-to-left text,
and for each prepended character $a$ to the text $T\$$,
the number of internal explicit nodes
from the leaf for $T$ to its lowest ancestor $u'$
for which hard W-link $\W{a}{u'}$ exists can be amortized constant.
This amortization argument is based on the fact that
the depth of the path from the root to the handle leaf $\ell'$
representing the extended text $a T \$$ is by at most one larger than
that of the path from the root to the handle leaf $\ell$ representing $T \$$.
This property holds also in the semi-online setting,
since while the $k$th text $T_k\$_k$ is being extended from right to left,
other texts remain static and thus do not change
the topology of the suffix tree.
Hence, we can build $\STree(\mathcal{T})$ for a collection
of semi-online left-to-right texts in $O(N \log \sigma)$ time and $O(N)$ space.

Blumer et al.~\cite{blumer85:_small_autom_recog_subwor_text} showed
how the DAWG for a collection $\mathcal{S}$ of semi-online left-to-right texts
can be built in $O(N \log \sigma)$ time.
Recall that each DAWG node represents an equivalence class of
substrings which have the same ending positions in the texts.
Appending a new character $a$ to the currently processed text $\$_k S_k$
can affect some equivalence class under the current text collection.
This can cause splitting an existing node into two nodes.
Let $w$ be the node that gets split
and $w'$ be the copy of this node $w$.
The original node $w$ will contains longer substrings than the copy $w'$.
The longest element belonging to $w'$ is the
\emph{longest repeating suffix} $X$
of $\$_k S_k a$ in the updated text collection,
and any element of $w$ that is shorter than $X$ will belong to $w'$.
Eventually, any element of $w$ that is longer than $X$ remains in $w$.
This node split operation can be done by redirecting
corresponding in-coming edges from $w$ to $w'$.
The key argument in the time analysis of Blumer et al.'s algorithm
is that this cost of redirecting in-coming edges
can also be amortized constant per added character $a$.
Observe that this update is exactly the same as
the above-mentioned update of the suffix tree
for the corresponding right-to-left text collection.
For instance, the longest repeating suffix of $\$_k S_k a$
for the current left-to-right text collection 
is the reverse of the longest repeating prefix of $a T_k \$_k$
for the corresponding right-to-left text collection.
Also, redirecting those in-coming edges in the DAWG
are exactly the same as updating a soft W-link to a hard one 
and redirecting soft W-links, in the suffix tree
of the corresponding right-to-left texts (recall Case (1) above).
Consequently, we can build the DAWG for a collection
of semi-online left-to-right texts in $O(N \log \sigma)$ time
and $O(N)$ space as well.

\subsection{Fully-online construction of Weiner's suffix trees and DAWGs}

In this subsection, we consider how to maintain the suffix tree
for a collection of $K$ texts which grow from right to left
in a fully-online manner.
This means that we will have to maintain $K$ handle leaves for the $K$ texts
simultaneously.
We also consider how to maintain the DAWG for a collection of $K$ texts
which grow from left to right in a fully-online manner.

Unfortunately, 
the identical amortization argument in both algorithms does not carry over 
in the fully-online setting.
However, we will show next that Weiner's algorithm can be modified
to work within the desired $O(N \log \sigma)$ time and $O(N)$ space bounds
with the aid of $\sigma$ nearest marked ancestor (NMA)
data structures of total size $O(N)$,
where $\sigma$ denotes the number of all distinct characters
appearing in the texts in the collection.
Moreover, the same data structures
can provide access to the DAWG edges, which cannot be maintained explicitly
within our bounds, 
in $O(\log \sigma)$ time per edge query.

We will use the following NMA data structure
as a building block of our algorithm.
\begin{lemma}[\cite{westbrook92:_fast_increm_planar_testin}] \label{lem:nma}
  There exists an NMA data structure
  for a growing rooted tree,
  which supports the following operations 
  in amortized $O(1)$ time each:
  1) find the NMA of a given node;
  2) insert an unmarked node;
  3) mark an unmarked node.
  This NMA data structure requires linear space
  in the size of the tree.
\end{lemma}

Suppose that we have $\STree(\mathcal{T}_{U[1..i-1]})$
for a fully-online right-to-left text collection $\mathcal{T}_{U[1..i-1]}$
and assume $U[i] = (k, a)$, i.e., the $k$th text $T_k\$_k$ gets
extended with a new character $a$ being prepended to it.
As in the case with the semi-online texts,
some new soft and hard W-links are created
in the updated $\STree(\mathcal{T}_{U[1..i]})$.
Fortunately, the number of such newly created W-links are bounded by the
size of the resulting suffix tree, which is $O(N)$.
However, the number of \emph{redirected} soft W-links,
which are the same as the number of DAWG edges to be redirected,
can be too numerous to be done 
within our desired bounds as the next lemma shows.

\begin{lemma} \label{lem:Weiner_lower_bound}
  Weiner's suffix tree algorithm 
	takes $\Theta(N \min(K, \sqrt N))$ time
  in the fully-online setting, where $N$ is the total length of the $K$ texts.
  Hence, for $K=\Theta(\sqrt N)$ it also takes $\Theta(N \sqrt{N})$ time
  to explicitly maintain the soft W-links
  (equivalently, the DAWG secondary edges)
  in the fully-online setting.
  The lower bound holds for a constant alphabet.
\end{lemma}

\begin{proof}
  To show that these bounds hold for constant alphabets,
  we here assume that each text in the collection terminates  
  with the same end-marker $\$$.
  However, in our collection of texts each text will be distinct,
  so that each $T_k \$$ will be represented by a unique handle leaf.
  
  First, we consider a lower bound.
  Consider the following $K$ right-to-left texts
  $\mathcal{T} = \{T_k = a^k \$ \mid 1 \leq k \leq K\}$
  where $a \in \Sigma$ and each text terminates with a common end-marker $\$$.
  Suppose we have constructed the suffix tree of $\mathcal{T}$ in any order.
  Then, we prepend a new character $c \in \Sigma$, such that 
  $c \neq a$, to each text $T_k=a^k\$$ in decreasing order of their length,
  $k = K, \ldots, 1$.
  Since we process each text in decreasing order of $k$,
  there are $\Omega(k)$ explicit nodes in the path
  from the handle leaf for $T_k=a^k\$$ to its lowest ancestor $r = \varepsilon$
  (the root) for which hard W-link $\W{c}{r}$ is defined.
  Hence, it takes $\Omega(k)$ time to na\"ively walk up this path.
  Also, with the exception of the first longest text $T_K$ that
	introduces $\Omega(k)$ new soft W-links, for all other $k<K$,
	there are $\Omega(k)$ soft W-links to be redirected along the way.
  Thus, there are $\Omega(K^2)$ edge re-directions in total, for all $k$'s.
  %
  We then repeat the above procedure several times.
  At each repetition $i$~($i > 1$),
  for each $k$ in decreasing order
  it again takes $\Omega(k)$ time to walk up from the handle leaf for $c^{i-1}a^k\$$
  until reaching its lowest ancestor $r$ for which hard W-link $\W{c}{r}$
  is defined.
  Also, there are $\Omega(k)$ soft W-links to be redirected along the way.
  Thus, at each repetition $i$,
  it takes a total of $\Omega(K^2)$ time for all $k$'s, too,

  Let $N$ be the total length of the texts in the collection
  after performing the above procedure several times.
  The initial total length of the text collection
  $\mathcal{T} = \{a^k \$ \mid 1 \leq k \leq K\}$ is $\frac{K(K+3)}{2}$.
  We then append $c$'s to each of the $K$ texts,
  and the text collection of total length finally becomes $N$.
  Hence, the number of iterations is $(N - \frac{K(K+3)}{2})/K = \Theta(N/K - K)$,
  which is $\Theta(N/K)$ in the case where $K < \alpha \sqrt{N}$ with
  some constant $\alpha$.
  Since each iteration requires re-directions of $\Omega(K^2)$ soft W-links,
  it takes a total of $\Omega(NK)$ time in this case.
  Now consider the case where $K > \alpha \sqrt{N}$.
  In this case, we can apply the same procedure as above
  only to $\alpha \sqrt{N}$ texts in the collection,
  and the other $K - \alpha \sqrt{N}$ texts remain empty.
  This leads to $\Omega(N \sqrt{N})$ total work for re-directing
  soft W-links.
  Combining these two,
  we obtain an $\Omega(N \min(K, \sqrt{N}))$ lower bound.
  

	To see that this lower bound actually gives rise to the worse case in Weiner's algorithm,
	we can focus only on the time required for soft W-link re-direction, 
	since new edge insertions and node insertions are always accounted globally to be the total size of the the suffix tree, which is $O(N)$.
	
	Recall that the number of soft W-link re-directions when appending a symbol $a$ to text $T_k \$$ is
	no larger than the suffix tree depth of the handle leaf representing $T_k \$$, which is in turn smaller
	than the length of $T_k \$$. 
	Also, the depth of the new leaf $a T_k \$$ is at most one more than
	the depth of leaf $T_k \$$ minus the number of edge re-directions that reduce depth
	of the current handle leaf associated with each of the $K$ text,
	while the depth of all current handle leaves $T_i \$$, $i \neq k$, 
	may also increase by at most one while updating $T_k \$$, by the insertion of the internal node
	off which the leaf $T_k \$$ is hanging above the handle leaf of $T_i$. 
	Thus, each of the $O(N)$ symbols may increase by at most one the depth of all the $K$ handle leaves.
	This depth increase was not an issue in the semi-online
	setting since previous $T_k \$$ are no longer updated and their handle leaves were no longer used.
	In the fully-online setting, this depth increase is problematic.
  The depth reduction argument gives an obvious $O(N)$ upper bound on the
	soft W-link re-directions while updating each of the $K$ texts, which adds up to $O(KN)$ overall upper bound.
	
	The analysis will separate those short texts $T_k \$$, such that $|T_k \$_k| \leq \sqrt N$ from the longer texts.
	For the short texts, each time a symbol is prepended to a text $T_k \$$, the number of soft W-link edge re-directions is bounded
	by the length of each short text, which is at most $\sqrt N$, totaling at most $O(N \sqrt N)$ such re-directions.
	For the long texts, we observe that there are at most $O(\sqrt N)$ such long texts, and for each specific
	text, the total number of soft W-link edge re-directions is at most $O(N)$, totaling at most $O(K N) \subseteq O(N \sqrt N)$. 
	Combining these bounds, we get the desired $\Theta(N \min(K, \sqrt N))$ tight bound.
\end{proof}

\begin{remark}
  To show that the bounds hold for a constant alphabet,
  we used the same end-marker $\$$ for all the texts
  in the proof of Lemma~\ref{lem:Weiner_lower_bound}. 
  We remark that the same arguments hold for
  the case where each text $T_k$ is terminated with a unique end-marker $\$_k$,
  as we assume elsewhere in this paper,
  since also in this case each text $T_k \$_k$ is represented by
  a unique handle text.
  We then use $K+2$ characters in the lower bound example
  (the alphabet is $\{a, b, \$_1, \ldots, \$_K\}$).
\end{remark}


To avoid the above-stated super-linear cost in Lemma~\ref{lem:Weiner_lower_bound},
we shall only maintain hard W-links and will not explicitly
maintain soft W-links.
Instead of soft W-links we will maintain only the Boolean indicator
$\sw_{a}(v)$ that tells us whether a (soft or hard) W-link $\W{a}{v}$ is defined
or not.
Once $\sw_{a}(v)$ is set to $1$,
it remains $1$ and does not need to be updated 
even when the corresponding soft W-link would have to be redirected.

\begin{figure}[t]
  \centerline{
    \includegraphics[width=1.0\linewidth]{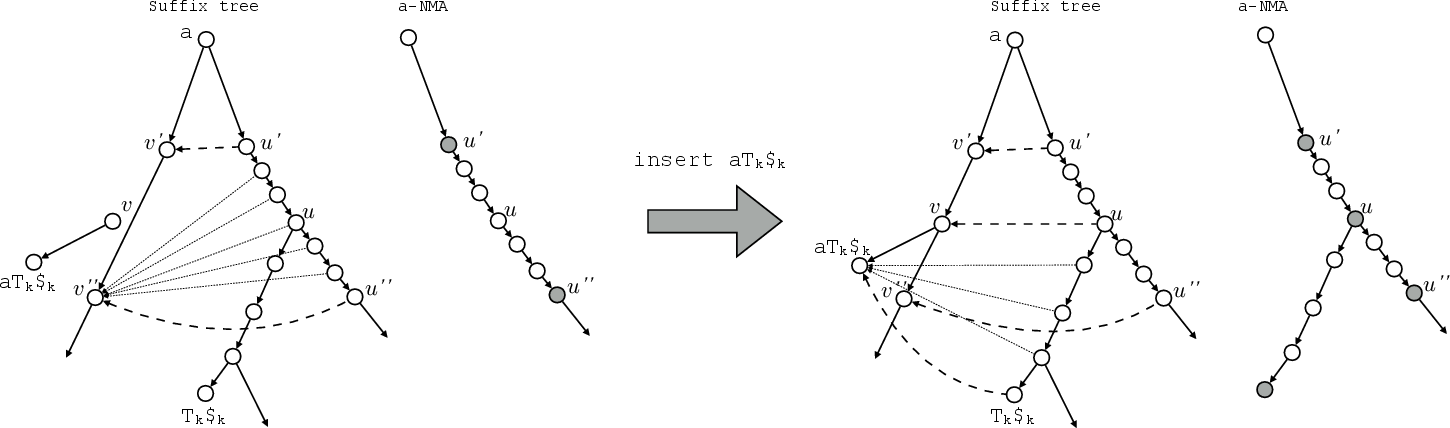}
  }
  \caption{
    Extending the text $T_k\$_k$ to $a T_k\$_k$.
    Soft W-links are shown short-dashed and hard W-links are shown long-dashed.
    Gray nodes of NMA data structures  mean marked nodes. 
    Nodes with new hard W-link (of added character $a$) in the suffix tree are added to the NMA as marked nodes.
    Nodes with new soft W-link (of $a$) are added to the NMA as normal nodes.
  }
  \label{fig:NMA}
\end{figure}

Like in the semi-online setting,
we here also go up from the leaf $\ell$ representing $T_k\$_k$
to its lowest ancestor $u'$ for which $\W{a}{u'}$ is defined.
The cost for walking up to the lowest ancestor $u$ of $\ell$
for which $\sw_{a}(u) = 1$ can be charged to the cost for
creating new soft W-links (or equivalently, that for creating new corresponding DAWG edges),
which is amortized constant per added character $a$.
One problem remains:
We would like to skip all explicit nodes $s'$ in the path from node $u$ to $u'$,
since na\"ively walking up this path can be as costly as
redirecting W-links $\W{a}{s'}$ for all such nodes $s'$.
In so doing, 
we shall also maintain for each character $\sigma$ an NMA data structure
of Lemma~\ref{lem:nma} on the subtree of the suffix tree which consists of
the two following disjoint sets of nodes:
(1) the set of unmarked nodes
$v$ such that $\sw_{a}(v) = 1$ and $\W{a}{v}$ is a soft W-link, and
(2) the set of marked nodes $v$ such that $\W{a}{v}$ is a hard W-link.
Our version of Weiner's algorithm will na\"ively
walk up the suffix tree from the leaf $\ell$ representing $T_k\$_k$
until the lowest node $u$ such that $\sw_{a}(u) = 1$,
and from there it will jump to $u'$ using the NMA data structure
for the prepended character $a$.
In what follows, 
we will denote this as the $a$-NMA data structure.

\begin{theorem} \label{theo:fully-online_weiner}
Given a fully-online sequence $U$ of $N$ update operators
for a collection of $K$ right-to-left texts $\mathcal{T}$,
our version of Weiner's algorithm can update the suffix tree
in a total of $O(N \log \sigma)$ time and $O(N)$ space.
\end{theorem}

\begin{proof}
  The correctness of our algorithm should be clear from the above discussion.

  Let us analyze the time complexity.
  The algorithm will now still climb up the suffix tree
  from the currently focused leaf $\ell$
  up to its lowest ancestor $u$ with $\sw_{a}(u) = 1$.
  From there, it would jump
  to its nearest ancestor $u'$ of $u$ having hard W-link $\W{a}{u'}$ defined
  in constant amortized time using an NMA query on the $a$-NMA data structure.
  Now we update the $a$-NMA data structure.
  If the insertion point $v$ for the new leaf $\ell'$ representing $aT_k\$_k$
  is newly created (see Case (1) in the previous sub-section),
  then the soft W-link $\W{a}{u}$ becomes hard.
  Hence, we mark node $u$ in the $a$-NMA data structure.
  Otherwise, the W-link $\W{a}{u}$ is already hard and hence
  $u$ is already marked in the $a$-NMA data structure.
  Recall that each node $s$ between the leaf $\ell$ and $u$
  obtain new soft W-links and hence $\sw_{a}(s)$ is now set to $1$.
  Hence, we insert an unmarked node for each $s$
  in the $a$-NMA data structure.
  Since the NMA data structure allows us to insert a new leaf
  in amortized constant time, we insert these unmarked nodes
  in increasing order of depth,
  from the child of $u$ to the parent of $\ell$ contained in the path.
  We also spend $O(\log \sigma)$ time
  at each visited node for searching the corresponding
  NMA data structure.
  Overall, it takes a total of $O(N \log \sigma)$ time
  to construct the suffix tree for a fully-online right-to-left text collection
  of total length $N$.

  Let us now analyze the space complexity.
  For each character $c \in \Sigma$,
  each marked node $u$ in the $c$-NMA data structure
  corresponds to a unique hard W-link $\W{c}{u}$.
  Also, each unmarked node $s$ in the $c$-NMA data structure
  corresponds to a unique soft W-link $\W{c}{s}$.
  Since the total number of hard and soft W-links for
  all characters $c \in \Sigma$ is $O(N)$,
  the total size of the $c$-NMA data structures for all characters
  $c \in \Sigma $ is $O(N)$.
\end{proof}

Now we turn our attention to construction of
the DAWG for a fully-online left-to-right text collection $\mathcal{S}$.
Since our version of Weiner's algorithm does not explicitly maintain
soft W-links, we do not have explicit representation of
secondary edges of the DAWG for the left-to-right texts.
However, the Weiner's suffix tree augmented with the NMA data structures
indeed is implicit representation of the DAWG secondary edges:

\begin{lemma} \label{lem:simulate_secondary}
  Using Weiner's suffix tree augmented with the NMA data structures,
  we can simulate each soft W-link per query in amortized
  $O(\log \sigma)$ time. 
\end{lemma}

\begin{proof}
  A given node $u$ has soft W-link $\W{a}{u}$
  for a given character $a$ iff $\sw_{a}{u} = 1$ and
  $\W{a}{u}$ is not a hard W-link.
  Suppose $u$ has soft W-link $\W{a}{u}$.
  We query the NMA $u'$ of $u$ in the $a$-NMA data structure.
  Let $b$ be the first character of the path label from $u'$ to $u$.
  We follow the hard W-link $\W{a}{u'} = v'$,
  and find the out-going edge of $v'$ whose edge label begins with $b$.
  Then, the child $v''$ of $v'$ below this edge is the
  destination of the soft W-link $\W{a}{u}$.
  The time for the NMA query is amortized to $O(1)$ and
  finding the appropriate $a$-NMA data structure and the
  appropriate out-going edge of $v'$ takes $O(\log \sigma)$ time each.
\end{proof}

The next corollary immediately follows from
Theorem~\ref{theo:fully-online_weiner} and Lemma~\ref{lem:simulate_secondary}.

\begin{corollary} \label{coro:fully-online_dawg}
Given a fully-online sequence $U$ of $N$ update operators
for a collection of $K$ left-to-right texts,
the DAWG can be maintained in a total of $O(N \log \sigma)$ time
and $O(N)$ space with $O(\log \sigma)$ query time 
for an out-going DAWG edge.
\end{corollary}


\section{Fully-online version of Ukkonen's suffix tree algorithm}
\label{sec:fully-online_suffix-tree}

Ukkonen's algorithm~\cite{Ukkonen95} constructs
the suffix tree of a given text in an online manner,
from left to right.
In this section, we show how Ukkonen's algorithm can be extended
to maintain the suffix tree for
a fully-online left-to-right text collection.
We will do so by first explaining that Ukkonen's algorithm can readily be
extended to the semi-online setting.
Then, we will describe some difficulties in
extending Ukkonen's algorithm to our fully-online setting,
and finally we will present how to overcome these difficulties
achieving $O(N \log \sigma)$-time algorithm.

\subsection{Semi-online left-to-right suffix tree construction}
Ukkonen's algorithm~\cite{Ukkonen95} can easily be extended to incrementally
construct the suffix tree for multiple texts in the semi-online setting.

Let $U$ be a semi-online sequence of $N$ update operators
such that the last update operator for each $k$~($1 \leq k \leq K$)
is $(k, \#_k)$, where $\#_k$ is a special end-marker for the $k$th text in the collection.
Also, assume that we have already constructed $\STree(\mathcal{S}_{U[1..i-1]})$
and that the next update operator is $U[i] = (k, a)$.
Thus a new character $a$ is appended to the text $S_k$ 
and it becomes $S_ka$.

In updating $\STree(\mathcal{S}_{U[1..i-1]})$ to $\STree(\mathcal{S}_{U[1..i]})$,
we have to assure that all suffixes of the extended text $S_ka$ will be represented by $\STree(\mathcal{S}_{U[1..i]})$.
These suffixes are categorized to three different types: 
\begin{description}
 \item[Type-1] The suffixes of $S_ka$ that are longer than $\longsuf_{\mathcal{S}_{U[1..i-1]}}(S_k)a$.
 \item[Type-2] The suffixes of $S_ka$ that are not longer than $\longsuf_{\mathcal{S}_{U[1..i-1]}}(S_k)a$
               and are longer than $\longsuf_{\mathcal{S}_U[1..i]}(S_ka)$.
 \item[Type-3] The suffixes of $S_ka$ that are not longer than $\longsuf_{\mathcal{S}_{U[1..i]}}(S_ka)$.
\end{description}

The suffixes of $S_ka$ are inserted in decreasing order of length.

The Type-1 suffixes are maintained as follows.
Let $s$ be any suffix of $S_k$ which is represented by a leaf of $\STree(\mathcal{S}_{U[1..i-1]})$.
Since $s$ is a non-repeating suffix of $S_k$ in $\mathcal{S}_{U[1..i-1]}$,
$sa$ is a non-repeating suffix of $S_ka$ in $\mathcal{S}_{U[1..i]}$,
which implies that $sa$ will also be a leaf of $\STree(\mathcal{S}_{U[1..i]})$.
Based on this observation, 
the label of the in-coming edge of the leaf is represented 
by a pair $\langle k, b \rangle$ called an \emph{open edge},
where $b$ is the beginning position of the label of the in-coming edge in 
the $k$th text.
We can retrieve the ending position of the edge label
in constant time by looking at the current length of the $k$th text.
This way, every existing leaf will then be ``automatically'' extended.

Hence, updating $\STree(\mathcal{S}_{U[1..i-1]})$ to $\STree(\mathcal{S}_{U[1..i]})$
reduces to inserting the Type-2 suffixes of $S_ka$
(note that the Type-3 suffixes of $S_k a$ already exists in the suffix tree).
For this sake, the algorithm maintains an invariant which indicates
the locus of $x = \longsuf_{\mathcal{S}_{U[1..i]}}(S_k)$ 
on $\STree(\mathcal{S}_{U[1..i-1]})$ called the \emph{active point}.
Since $x$ can be an implicit node,
the algorithm maintains the canonical reference $(v, c, \ell)$ to $x$.
For convenience, if $x$ is an explicit node,
then let its canonical reference be $(x, \varepsilon, 0)$.
The update starts from the current active point $x$
represented by its canonical reference pair,
and the Type-2 suffixes of $S_ka$ are inserted in decreasing order of length,
by using the chain of (virtual) suffix links. There are two cases:
\begin{enumerate}
 \item[I.] If it is possible to go down from $x$ 
 with character $a$, then no updates to the tree topology are needed.
 The new active point is $xa$, and the reference to $xa$ is made canonical
 if necessary. The update ends.

 \item[II.] If it is impossible to go down from $x$
 with character $a$, then we create a new leaf.
 Let $j$ be the beginning position of the suffix of $S_ka$
 which corresponds to this new leaf.
 The following procedure is repeated until Case I happens.

  \begin{enumerate}
   \item If the active point $x$ is on an explicit node, then a new 
   leaf node $s$ is created as a new child of $x$,
   with its incoming edge labeled by $\langle k, b \rangle$,
   where $b = |S_ka| - |x|+1$.
   The active point $x$ is updated to $\suflink(x)$.

   \item If the active point $x$ is on an implicit node, 
   then $x$ becomes explicit in this step.
   A new leaf node $s$ is created as a new child of 
   $x$ with its incoming edge labeled by $\langle k, b \rangle$.
   Since the suffix link of the new explicit node $x$ does not yet exist,
   we simulate the suffix link traversal as follows:
   Let $(v_j, c_j, \ell_j)$ be the canonical reference to $x$.
   First, we follow the suffix link $\suflink(v_j)$ of $v_j$,
   and then go down along the path of length $\ell_j$ 
   from $\suflink(v_j)$ starting with character $c_j$.
   Let this locus be $x'$.
   Let $v_{j+1}$ be the longest explicit node in this path.
   (i) If $|v_{j+1}| = |x'|$, then we firstly create 
   the new suffix link $\suflink(x) = v_{j+1}$ for the new explicit node $x$.
   The active point $x$ is updated to $x'$ 
   and is represented by canonical reference $(v_{j+1}, \varepsilon, 0)$.
   (ii) If $|v_{j+1}| < |x'|$, then the next active point is implicit.
   The active point $x$ is updated to $x'$  
   and is represented by canonical reference $(v_{j+1}, c_{j+1}, \ell_{j+1})$.
   The suffix link of $x$ will be set to $x'$
   when $x'$ becomes explicit in the next step.
  \end{enumerate}
\end{enumerate}
The most expensive case is II-b-(ii).
Since the path from $v_{j+1}$ to $x'$ contains at most $\ell_j - \ell_{j+1}$ 
explicit nodes, 
it takes $O((\ell_j-\ell_{j+1}+1)\log \sigma)$ time 
to locate the next active point $x'$ (note $\ell_j - \ell_{j+1} \geq 0$ holds).
All the other operations take $O(\log \sigma)$ time.
Hence, the total cost to insert all leaves (suffixes) for 
the $k$th text is 
$O(\sum_{j=1}^{N_k}(\ell_j - \ell_{j+1} + 1) \log \sigma) = O(N_k \log \sigma)$,
where $N_k$ is the final length of the $k$th text.
Thus the amortized time cost for each leaf (suffix) for the $k$th text 
is $O(\log \sigma)$.
Overall, it takes a total of $O(N \log \sigma)$ time to construct
$\STree(\mathcal{S}_U)$ for a semi-online sequence $U$ of
update operators. The space requirement is $O(N)$.

\subsection{Difficulties in fully-online left-to-right suffix tree construction}

The following observations suggest
that it does not seem easy to extend Ukkonen's algorithm to our left-to-right fully-online setting:
\begin{enumerate}
 \item[A.]
 \textbf{[Keeping track of active points]} Let $U[i] = (k, a)$ which updates the current $k$th text $S_k$ to $S_ka$,
 and assume that we have just constructed $\STree(\mathcal{S}_{U[1..i]})$.
 Recall that we defined the initial locus of the active point for $S_ka$ 
 on $\STree(\mathcal{S}_{U[1..i]})$ to be 
 the longest repeating suffix of $T_ka$ in $\mathcal{S}_{U[1..i]}$.
 However, since $U$ is fully-online,
 any other text $T_h$~($h \neq k$) in the collection
 would be updated by following update operators $U[r]$ with $r > i$.
 Then, the longest repeating suffix of $S_ka$ in $\mathcal{S}_{U[1..r]}$ 
 can be much longer than that of $S_ka$ in $\mathcal{S}_{U[1..i]}$. 
 In other words, some Type-1 suffixes of $S_ka$ in $\mathcal{S}_{U[1..i]}$ 
 can become of Type-2 in $\mathcal{S}_{U[1..r]}$. 
 What is worse, updating $S_h$ can affect the longest repeating suffix 
 of any other text in the collection as well.

 \item[B.]
 \textbf{[Canonization of active points]}
 Even if we somehow manage to efficiently maintain the active point for each text 
 in the collection, there remains another difficulty.
 Let $j$ be the beginning position of the longest repeating suffix of $S_ka$ 
 in $\mathcal{S}_{U[1..i]}$,
 and let $(v_j, c_j, \ell_j)$ be the canonical reference to this suffix.
 Let $U[i'] = (k, a')$ be the first update operator in $U$ which 
 updates the $k$th text after $U[i] = (k, a)$.
 Let $(v'_j, c'_j, \ell'_j)$ be the canonical reference
 to the longest repeating suffix of $S_ka$ in $\mathcal{S}_{U[1..i']}$,
 which is the ``real'' initial active point where insertion of
 the Type-2 suffixes should start at this $i'$th step.
 By the property of suffix trees $\ell'_j \geq \ell_j$ holds,
 and what is worse, this length $\ell'_j$ 
 is unbounded by the number of Type-2 suffixes inserted at this $i'$th step.
 Thus, it is not clear whether
 the amortization technique we used for the semi-online construction
 works in our fully-online setting.

 \item[C.]
 \textbf{[Maintaining leaf ownerships]}
 The phenomenon mentioned in Difficulty A also causes
 a problem of how to represent the labels of the in-coming edges to the leaves.
 Assume that we created a new leaf w.r.t. an update operator $(k, a)$,
 and let $\langle k, b_k \rangle$ be the pair representing 
 the label of the in-coming edge to the leaf,
 where $b_k$ is the beginning position of the edge label in the $k$th text.
 We say that the $k$th text $S_k$ is the \emph{owner} of the leaf.
 It corresponds to a Type-1 suffix of the $k$th text, 
 but the leaf can later be extended by another growing text $S_h$.
 Namely, $S_h$ can overtake the ownership of the leaf from $S_k$.
 After this happens, then the pair $\langle k, b_k \rangle$
 has to be updated to $\langle h, b_h \rangle$,
 where $b_h$ is the beginning position of the edge label in the $h$th text.
 Notice that this update may happen repeatedly.
\end{enumerate}

\subsection{Fully-online left-to-right suffix tree algorithms}

Let us now consider how to construct the suffix tree
for a fully-online left-to-right text collection.
Our fully-online version of Ukkonen's algorithm
works with the aid of the fully-online version of Weiner's algorithm
proposed in Section~\ref{sec:weiner}.
Namely, for a fully-online left-to-right text collection
$\mathcal{S}$ with $K$ texts,
we build $\STree(\mathcal{S})$ in tandem with $\STree(\mathcal{T})$,
where $\mathcal{T}$ is the set of reversed texts from $\mathcal{S}$
(i.e., $\mathcal{T} = \rev{\mathcal{S}}$).
Since we use the fully-online version of Weiner's algorithm,
as in Section~\ref{sec:weiner}, we assume that
each text in $\mathcal{T}$ terminates with a special symbol $\$_k$,
namely, $\mathcal{T} = \{T_1 \$_1, \ldots, T_K \$_K\}$.
This in turn implies that each text in $\mathcal{S}$ begins with $\$_k$,
namely, $\mathcal{S} = \{\$_1 S_1, \ldots, \$_K S_K\}$,
where $S_i = \rev{T_i}$ for $1 \leq i \leq K$.

In what follows, we will propose two alternative approaches.
Suppose we have constructed $\STree(S_{U[1..i-1]})$.
Given the $i$th update operator $U[i] = (k, a)$,
the first one called the \emph{forward approach}
traverses a chain of (virtual) suffix links in a forward manner
and inserts new leaves of the updated text $\$_k S_k a$
in decreasing order of the lengths of the suffixes of $\$_k S_k a$.
This forward approach is a direct extension of Ukkonen's original algorithm.
The second one called the \emph{backward approach}
traverses a chain of (virtual) suffix links in a backward manner and
inserts new leaves in increasing order of the lengths of the suffixes of $\$_k S_k a$.
This backward approach can be seen an extension of
Breslauer and Italiano's algorithm~\cite{BreslauerI13}
which was originally proposed for real-time suffix tree construction
for a single left-to-right text.

\subsubsection{Forward approach} \label{sec:forward}

In this subsection, we present our forward approach
to update $\STree(S_{U[1..i-1]})$ to $\STree(S_{U[1..i]})$.
The key notions in this forward approach are
\emph{swapping active points} and
\emph{tight connections between active points and leaf ownerships}.
In what follows we will explain these notions in full details.

Let us first consider maintaining active points (Point A).
This is indeed closely related to maintaining leaf ownerships (Point C).
We will for now put it aside the cost for maintaining leaf ownerships,
and will focus on describing how active points can affect
ownerships of leaves.

For a single right-to-left online text,
the suffix links of the leaves form a single path
from the longest leaf to the shortest one.
On top of them
we also consider a virtual suffix link from the shortest leaf to the active point.

We generalize the above notion to our fully-online text collection $\mathcal{S}$.
Unlike the single text case, a leaf can represent a suffix of multiple texts
in our fully-online setting.
This implies that the suffix links of $\STree(\mathcal{S})$ form a forest.
Let $F_{\mathcal{S}}$ denote this forest.
This forest is only conceptual,
namely, in our algorithms to follow we will \emph{not} explicitly maintain it.
However, the forest gives us more insights into Points A and C.
Formally,
the forest $F_{\mathcal{S}}$ is a set of maximal trees 
such that each maximal tree $\SLT$ in $F_{\mathcal{S}}$ satisfies:
\begin{itemize}
  \item the root of $\SLT$ is the locus (an implicit or an explicit internal node) of the active point of a text, 
  \item the other nodes of $\SLT$ are leaves of $\STree(\mathcal{S})$, and
  \item the (reversed) edges of $\SLT$ are suffix links of $\STree(\mathcal{S})$ (if the root of $\SLT$ is an implicit node, then the (reversed) edges from the root to its children are virtual suffix links from the children).
\end{itemize}

Since a leaf of $\STree(\mathcal{S})$ can be a suffix of
multiple texts, there are multiple choices for the owner of each leaf.
Our choice of the owner of a leaf is either
\begin{enumerate}
  \item[(R1)] the text that created the leaf, or
  \item[(R2)] the last text whose active point has extended the leaf.
\end{enumerate}
Regarding Rule (R2) above,
we will soon describe in more details how the active point of a text
can extend an existing leaf.


Suppose that we have constructed $\STree(\mathcal{S}_{U[1..i-1]})$
and that we are given an update operator $U[i] = (k, a)$
which appends new character $a$ to text $\$_k S_k$.

If the active point of $\$_kS_k$  is not on
a leaf of $\STree(\mathcal{S}_{U[1..i-1]})$,
then the suffix tree is updated as in the semi-online setting
and there are no changes on the ownerships of the leaves.
Hence, in what follows we consider the case where
the active point of $\$_k S_k$ is on a leaf of $\STree(\mathcal{S}_{U[1..i-1]})$.

Let $s$ be the leaf of $\STree(\mathcal{S}_{U[1..i-1]})$
where the active point of the text $\$_k S_k$ lies.
Let $\SLT$ denote the suffix link tree in $F_{S_{U[1..i-1]}}$
that contains this node $s$,
and let $P_i$ be the path from $s$ to the root of $\SLT$.
Also, let $\mathcal{O}_i$ be the set of texts which are the owners of
the suffix tree leaves in $P_i$.
Finally, let $L_i$ be the list of all nodes $u$ in the path
from the \emph{parent} of $s$ to the root of $\SLT$
such that the active point of some text in $\mathcal{O}_i$ lies on $u$.
For each $1 \leq x \leq m = |L_i|$, let $u_x = L_i[x]$.
For convenience, let $u_0 = s$.
For each $1 \leq x \leq m$, let $k_x$ denote the text id of the owner of $u_x$.
Then, due to the way how the ownerships of leaves are defined
by Rules (R1) and (R2) above,
for every $1 \leq j < m$ the owner of every leaf
between $u_{j-1}$ and $u_{j}$ is the $k_{j}$th text in the collection.
See also the left diagram of Figure~\ref{fig:swap_active_points} for illustration.

Now we describe how the ownerships of leaves
and the active points of texts can change
when a new character is appended to a text in the fully-online setting.
We begin with the first node $u_1 = s$ in the list $L_i$
whose current owner is text $\$_{k_1}S_{k_1}$. 
See also the left diagram of Figure~\ref{fig:swap_active_points}.
Since $\$_{k} S_{k}$ now gets extended to $\$_{k} S_{k} a$,
the active point of this text \emph{extends} the suffix tree leaf $u_1$.
Then, the extended leaf $u_1$ no more represents a suffix of
its original owner $\$_{k_1} S_{k_1}$.
This implies that the new owner of this suffix tree leaf $u_1$ is $\$_k S_k a$.
The same happens to all leaves in the path up to $u_1$.
Then, we \emph{swap} the active points of texts $\$_kS_k$ and $\$_{k_1}S_{k_1}$.
We continue the same procedure recursively
for the other nodes $u_2$, \ldots, $u_m$ in the list $L_i$,
and finally the new owner of each leaf in the path $P_i$ becomes
the updated $k$th text $\$_kS_ka$.
After reaching the root of $\SLT$,
we possibly create new edges labeled with $a$
following virtual suffix links,
and finally arrive at the new locus of the active point for
the updated $k$th text $\$_kS_ka$.
This operation may split the original suffix link tree $\SLT$
into some smaller suffix link trees (see also Figure~\ref{fig:swap_active_points}).

\begin{lemma}
  The above procedure correctly maintains
  the active points of texts
  and the leaf ownerships under Rules (R1) and (R2).
\end{lemma}

\begin{proof}
  It is clear that the above procedure
  correctly maintains the leaf ownerships under Rules (R1) and (R2).

  Let $\$_{k_j} S_{k_j}$ be any text in $\mathcal{O}_i$.
  After swapping the active points of $\$_k S_k a$
  and those texts in $\mathcal{O}_i$,
  the locus of the active point of $\$_{k_j} S_{k_j}$
  is one character above the suffix tree leaf (say $u$)
  that has just been extended by $\$_k S_k a$.
  By the definition of list $L_i$,
  this leaf $u$ \emph{before extension} was the longest leaf whose previous
  owner was $\$_{k_j} S_{k_j}$.
  Hence, the string depth of the new active point of $\$_{k_j} S_{k_j}$
  is at least $|u|-1$.
  Also, it cannot be larger than $|u|-1$,
  since otherwise
  it contradicts with the definitions of $\mathcal{O}_i$ and $L_i$
  (see also Figure~\ref{fig:swap_active_points}).
  Hence, the above procedure of swapping active points
  correctly maintains the active points of the texts in the collection.
\end{proof}

\begin{figure}[h!t]
  \centerline{
    \includegraphics[scale=0.31, angle=270]{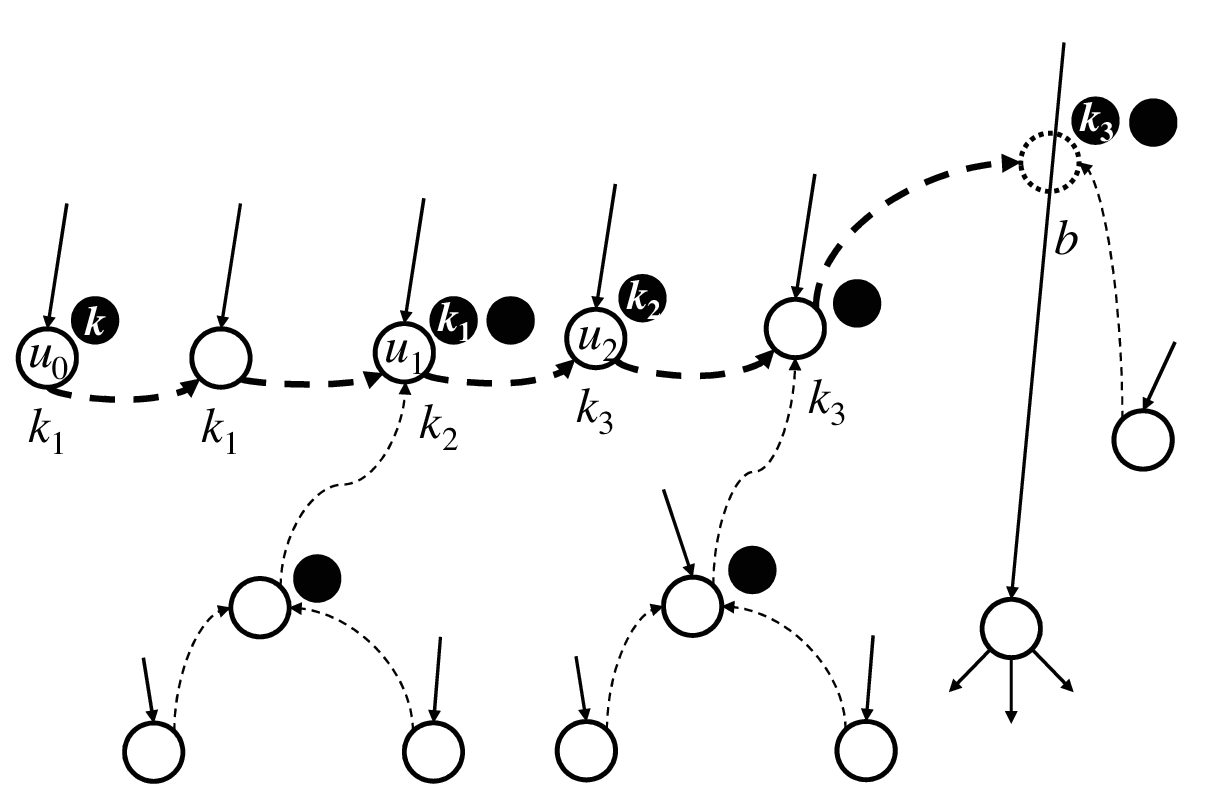}
    \hfill
    \includegraphics[scale=0.31, angle=270]{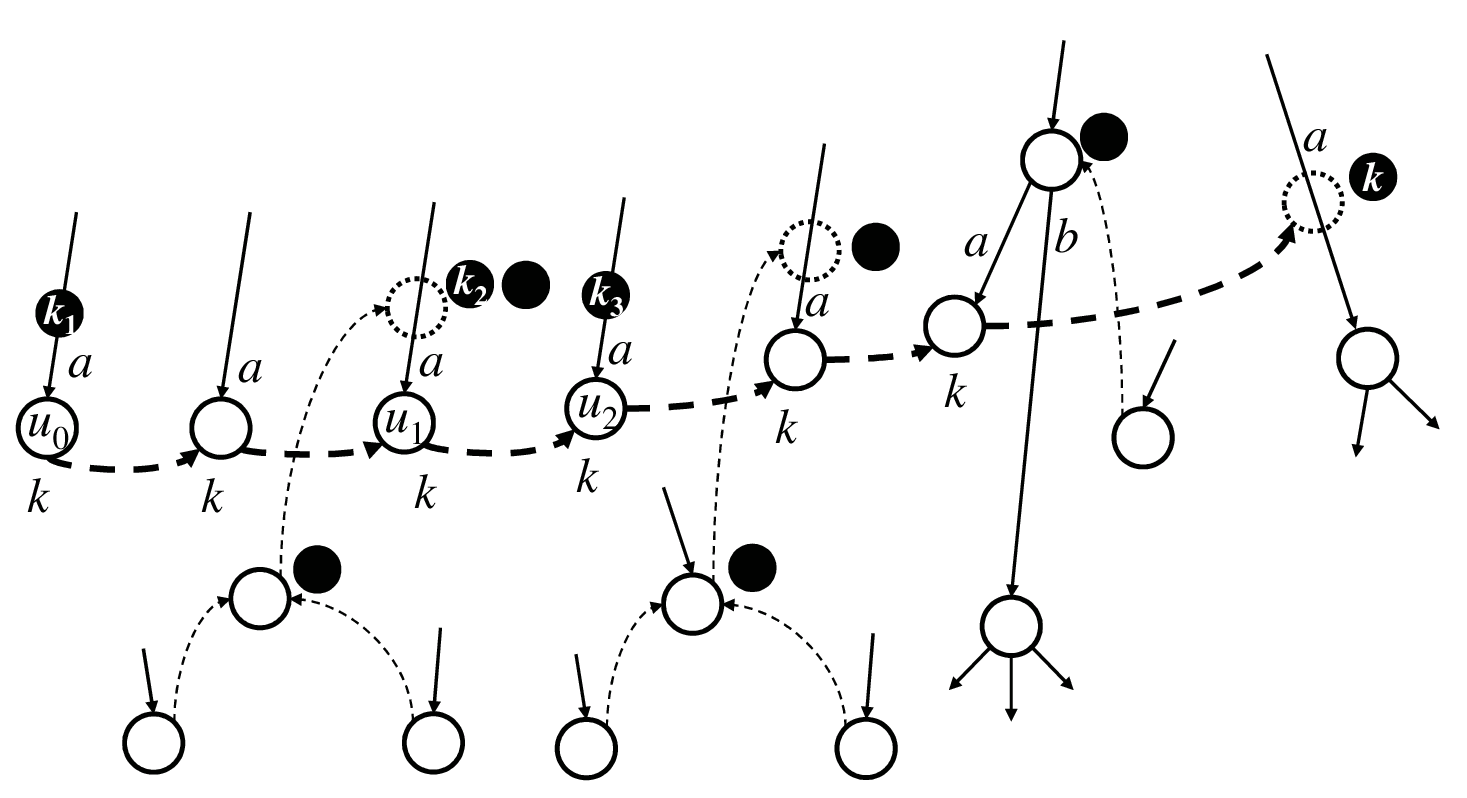}
  }
  \vspace{-0.2cm}
  \caption{The left diagram depicts a suffix link tree in the forest before the $i$th update with update operator $U[i] = (k, a)$, where the solid and broken arrows respectively represent suffix tree edges and suffix links, and the white and black circles respectively represent suffix tree leaves and active points. The dotted circle represents the root of the suffix link tree (which is an implicit node in this case). The suffix link path $P_i$ of interest is shown with bold broken arrows, where the staring node is $s = u_0$. The integers $k_j$ below each leaf shows the current owner of the leaf, and hence $\mathcal{O}_i = \{\$_k S_k, \$_{k_1}S_{k_1} , \$_{k_2}S_{k_2}, \$_{k_3}S_{k_3}\}$. The integer $k_j$ in each black circle implies that it is the active point for text $\$_{k_j}S_{k_j}$. The black circles without text id's are the active points of texts which are not in $\mathcal{O}_i$. The right diagram shows how it looks after the text $\$_k S_k$ has been extended to $\$_k S_k a$ with a new character $a$. Since its active point has extended the leaf $u_0$ with $a$, text $\$_{k}S_{k}$ becomes the new owner of every leaf in the path $P_i$. In the meantime, we swap the active point for text $\$_kS_k a$ with the active points of texts in $\mathcal{O}_i$, in the order they appear in the path $P_i$. After the active point of text $\$_k S_k a$ and that of the last text in the path (which in this figure is $\$_{k_3} S_{k_3}$) have been swapped, we possibly create new leaves (in this figure we create just one new leaf), and eventually we find the new locus for the active point for the updated text $\$_k S_k a$. Since all the leaves in the path $P_i$ have been extended by the new character $a$, this path breaks away from the original suffix link tree. As a result, we obtain several smaller suffix link trees.}
  \label{fig:swap_active_points}
\end{figure}

Now we wish to maintain leaf ownerships
as described above.
However, the next lemma shows that
it requires super-linear cost to explicitly maintain leaf ownerships.

\begin{lemma} \label{lem:leaf_ownership_lowerbound}
  There is a left-to-right fully-online collection
  of $K$ texts of total length $N$
  for which explicitly maintaining leaf ownerships requires
  $\Omega(\frac{N^2}{K})$ time.
\end{lemma}

\begin{proof}
  Consider an initial text collection $\mathcal{S} = \{\$_1, \ldots, \$_K\}$.
  We will update this text collection in $i$ rounds
  so that in each $j$th round
  the same character $a_j$ is appended to each text.
  The order of the texts to which $a_j$ is appended is arbitrary in each round.
  Thus, after the $j$th round,
  the text collection becomes of form
  $\{\$_1a_1 \cdots a_j, \ldots, \$_K a_1 \cdots a_j\}$.
  We also assume that $a_j \neq a_h$ for any $1 \leq j \neq h \leq i$.
  This implies that in each $j$th round,
  we will have $j$ leaves representing common suffixes
  $a_1 \cdots a_j$, $a_2 \cdots a_j$, \ldots, $a_j$.  

  Notice that during the $j$th round,
  the ownership of each such leaf has to be updated $K$ times
  since each such leaf is shared by the $K$ texts.
  Therefore, the total number of updates for the leaf ownership
  after the final $i$th round is at least
  \begin{equation}
    K(1 + 2 + \cdots + i) = \frac{Ki(i+1)}{2}. \label{eqn:leaf_ownership}
  \end{equation}  
  Since $N$ is the total length of the resulting text collection
  after the $i$th round, we get $N = K(i+1)$.
  Hence, $i = \Theta(\frac{N}{K})$.
  Plugging this into equation~\ref{eqn:leaf_ownership},
  we obtain the desired lower bound $\Omega(\frac{N^2}{K})$.
\end{proof}

The above $\Omega(\frac{N^2}{K})$ lower bound requires us
a super-linear cost for explicit leaf ownership maintenance
when $K = o(N)$.
Indeed, $K = o(N)$ is the only meaningful case in our
fully-online problem:
If $K = \Theta(N)$,
then each of the $K$ texts is of constant size
and hence a na\"ive algorithm would update
the suffix tree in constant time per each text
no matter how they are updated,
resulting in an $O(N)$-time construction anyway.
Hence, in what follows, we will only consider the case
where $K = o(N)$.

Due to Lemma~\ref{lem:leaf_ownership_lowerbound},
we shall not explicitly maintain leaf ownerships
in our fully-online algorithm.
However, when swapping the active point of the
$k$th text with those of the texts in the set $\mathcal{O}_i$,
we need to know the owner of the leaf that has just been extended
by the active point of the $k$th text.
We also need to know the set $\mathcal{O}_i$ of texts
which are the owners of the leaves in the path $P_i$,
and need to know the list $L_i$ of leaves
where those active points currently lie.
For this sake we use the aid of our version of
Weiner's algorithm for fully-online right-to-left construction.
Namely, we build $\STree(\mathcal{S}_{U[1..i]})$
in tandem with $\STree(\mathcal{T}_{U[1..i]})$
for each increasing $i = 1, \ldots, N$.
For simplicity, we will call the left-to-right fully-online suffix tree
$\STree(\mathcal{S}_{U[1..i]})$ as the \emph{Ukkonen tree}
and the right-to-left fully-online suffix tree $\STree(\mathcal{T}_{U[1..i]})$
as the \emph{Weiner tree}.

Below we show key observations that connect
our versions of Weiner's algorithm and Ukkonen's algorithm
in the fully-online setting.
For each node $v$ of the Weiner tree,
let $\wdeg(v)$ denote the number of (soft or hard) W-links from $v$,
namely, $\wdeg(v) = |\{c \in \Sigma \mid \sw_c(v) = 1\}|$.

\begin{lemma} \label{lem:connecting_Weiner_Ukkonen}
  Let $u$ be any leaf in the list $L_i$ of the Ukkonen tree
  $\STree(\mathcal{S}_{U[1..i-1]})$.
  Then, there exists an explicit node $v$ of the Weiner tree
  $\STree(\mathcal{T}_{U[1..i-1]})$ such that
  (1) $v = \rev{u}$,
  (2) $v$ is in the path from the root to the leaf representing $T_k\$_k$, and
  (3) $\wdeg(v) = 0$.
\end{lemma}

\begin{proof}
  Since $u$ is a leaf of the Ukkonen tree $\STree(\mathcal{S}_{U[1..i-1]})$,
  it is a suffix of the text $\$_k S_k$ to which a new character $a$ will be appended.
  Hence $v = \rev{u}$ is a prefix of the reversed text $T_k\$_k$,
  and is located on the path from the root to the leaf $T_k\$_k$
  in the Weiner tree $\STree(\mathcal{T}_{U[1..i-1]})$.
  By the definition of the list $L_i$,
  the active point of some other text (say $\$_h S_h$, with $h \neq k$)
  lies on the leaf $u$ in the Ukkonen tree,
  which implies that $u$ is the longest suffix of $\$_h S_h$
  that occurs at least twice in the left-to-right collection.
  Since each left-to-right text begins with a distinct $\$$ symbol,
  there must be at least two distinct characters that immediately
  precede occurrences of $u$.
  This in turn implies that there are at least two distinct characters
  that immediately follow occurrences of $v = \rev{u}$
  in the right-to-left text collection,
  and hence $v = \rev{u}$ is an explicit node in the Weiner tree.
  To prove (3) assume on the contrary that $\wdeg(v) > 0$,
  and let $c$ be any character such that $\sw_c(v) = 1$.
  Since $cv = c\rev{u}$ is a substring of some text
  in the right-to-left collection $\mathcal{T}_{U[1..i-1]}$,
  $uc$ is a substring of some text in the left-to-right collection
  $\mathcal{S}_{U[1..i-1]}$.
  However, this contradicts that $u$ is a leaf of
  the Ukkonen tree $\STree(\mathcal{S}_{U[1..i-1]})$.
  Hence $\wdeg(v) = 0$.
\end{proof}

As was shown in Section~\ref{sec:weiner},
when we update the Weiner tree $\STree(\mathcal{T}_{U[1..i-1]})$
to $\STree(\mathcal{T}_{U[1..i]})$
with update operator $U[i] = (k, a)$ which prepends character $a$ to text $T_k \$_k$,
we walk up from the leaf $T_k \$_k$ until
finding the first node with a (soft or hard) W-link w.r.t. $a$ defined.
Since the total cost of walking up these paths for all characters
prepended to the right-to-left texts
is linear in the final total length $N$ of all texts,
the number of nodes in the list $L_i$ for $1 \leq i \leq N$
is also linear in $N$.

Notice that not every explicit node $v$ with $\wdeg(v) = 0$
in the path from the leaf $T_k \$_k$ to the root of the Weiner tree
corresponds to a leaf in the list $L_i$ on the Ukkonen tree.
However, as was shown above,
we can afford to check each such explicit node $v$ in total linear time.

The next lemma shows how to maintain correspondence
between these nodes in the Weiner tree and the Ukkonen tree.

\begin{lemma} \label{lem:maintain_correspondence}
We can maintain correspondence between
each node $v$ of the Weiner tree with $\wdeg(v) = 0$
and its corresponding leaf $u$ in the Ukkonen tree
in $O(N \log \sigma)$ total time.
\end{lemma}

\begin{proof}
  Let $v$ be any node of the Weiner tree 
  $\STree(\mathcal{T}_{U[1..i-1]})$ with $\wdeg(v) = 0$.
  Suppose we have maintained correspondence between
  $v$ and its corresponding leaf $u$ in the Ukkonen tree
  $\STree(\mathcal{S}_{U[1..i-1]})$.
  This correspondence is maintained by bidirectional links
  between the two trees.
  
  Now suppose we are given an update operator
  $U[i] = (k, a)$ which appends a new character $a$ to $\$_k S_k$
  and prepends $a$ to $T_k \$_k$.
  There are three cases to consider.

  \begin{itemize}
  \item[(a)]
    If the active point of the $k$th left-to-right text
    extends a leaf of the Ukkonen tree:
    In this case,
    as was described previously and was illustrated in Figure~\ref{fig:swap_active_points},
    the leaves in the path $P_i$ get extended
    by the new character $a$ that was
    appended to the $k$th left-to-right text $\$_k S_k$.
    This implies that $v$ in the updated Weiner tree
    $\STree(\mathcal{T}_{U[1..i]})$ does not correspond to
    a leaf in the updated Ukkonen tree $\STree(\mathcal{S}_{U[1..i]})$.
    Thus, we remove the bidirectional link that connects $v$ and
    the corresponding leaf in the Ukkonen tree.

  \item[(b)]
    If the active point of the $k$th text catches up a leaf $u$ of the Ukkonen tree:
    Since $u$ is a leaf whose current
    owner is another text $\$_h S_h$ with $h \neq k$,
    $u$ is a suffix of at least two distinct left-to-right texts
    in the updated collection $\mathcal{S}_{U[1..i]}$.
    Hence, 
    $\rev{u}$ is a prefix of at least two distinct right-to-left texts
    in the updated collection $\mathcal{T}_{U[1..i]}$,
    and hence is represented by an explicit node in the
    updated Weiner tree $\STree(\mathcal{T}_{U[1..i]})$.
    Let $v$ be this explicit node.
    Moreover, since $u$ is the locus of the active point of $\$_k S_k a$,
    $u$ is the longest repeating suffix of $\$_k S_k a$
    and hence $v = \rev{u}$ is the longest repeating prefix of $a T_k \$_k$.
    This node $v$ is exactly the insertion point of
    the new leaf $a T_k \$_k$ in the Weiner tree.
    Hence, we can find the locus of $v = \rev{u}$
    during the updates of the Weiner tree
    and can easily create a bidirectional link between $v$ and $u$.

  \item[(c)] Otherwise, there are no changes in the correspondence
    and hence no maintenance of bidirectional links is needed.
  \end{itemize}

  In both cases (a) and (b), the costs can be charged to
  the construction of the Weiner tree which takes total $O(N \log \sigma)$ time.
\end{proof}

In Lemmas~\ref{lem:connecting_Weiner_Ukkonen} and~\ref{lem:maintain_correspondence} we have shown how to efficiently find those suffix tree leaves
in the list $L_i$ of the Ukkonen tree with the aid of the Weiner tree.
What remains is how to find each text in the set $\mathcal{O}_i$
of owners of the leaves in the list $L_i$.
The next lemma shows yet another application of the Weiner tree
for this purpose.

\begin{lemma} \label{lem:leaf_ownership_Ukkonen}
With the aid of the Weiner tree,
we can find the owner of each leaf in the list $L_i$
in total $O(N \log \sigma)$ time for all $1 \leq i \leq N$.
\end{lemma}

\begin{proof}
  In each internal explicit node of the Weiner tree,
  we store the id of the text which created
  the oldest leaf in the subtree rooted at this internal explicit node.
  This can be easily maintained in $O(1)$ time per node:
  When we split an edge and create a new internal node,
  then we simply copy the text id stored in its unique child.

  Consider any update operator $U[i] = (k, a)$.
  Let $u$ be any leaf in the list $L_i$ of the Ukkonen tree
  and let $v$ be its corresponding node in the Weiner tree
  (hence $v = \rev{u}$ and it is an explicit node due to Lemma~\ref{lem:connecting_Weiner_Ukkonen}).
  Then, if the text id stored in $v$ is $h$,
  then the $h$th text is the current owner of the leaf $u$ in the Ukkonen tree.
  This is true in either case
  where the leaf $u$ was created by the $h$th text and
  has never been extended by an active point,
  or the leaf $u$ was last extended by the $h$th text.
  In both cases,
  the subtree rooted at $v = \rev{u}$ in the Weiner tree
  may contain leaves which correspond to suffixes of some other texts
  than the $h$th text,
  but in the Ukkonen tree
  the active points of these texts only caught up with the leaf $u$.
  Hence none of these texts is the one which created the leaf $u$,
  or the last one that has extended $u$.
  Therefore, the $h$th text is the current owner of $u$.


  A careful consideration is required when the leaf $u$ gets extended
  by the active point of text $\$_k S_k$.
  Now the extended leaf represents the extended string $ua$
  and its new owner is the $k$th text $\$_k S_k a$.
  As was shown in the proof for Lemma~\ref{lem:maintain_correspondence},
  in the Weiner tree the reversed extended string $a\rev{u}$
  is represented by a new,
  different locus than the locus for $\rev{u}$.
  It is also possible that $a\rev{u}$ is on an implicit node in the Weiner tree
  at this stage, but it will become explicit when the active point
  of another text catches up the leaf $ua$ in the Ukkonen tree.
  Thus, we will be able to return the text id $k$
  as the correct answer for a leaf ownership query when the active point
  of another text extends the leaf $ua$ in future.
\end{proof}

In the above arguments we have shown that Difficulties A and C can be efficiently
resolved by swapping active points and by neglecting explicit maintenance of leaf ownerships.

Meanwhile, this lazy maintenance of leaf ownership causes
two more issues;
Suppose that the active point of some text $\$_i S_i$ lies on an edge
that leads to a leaf $u$, and that a new character $a$ has been appended
to this text.
Let $x$ be the string represented by the active point.
\begin{itemize}
\item The first question is how we can determine whether
the active point can step forward along this edge by character $a$,
or a new explicit node must be created at the locus of $x$
together with a new edge labeled with $a$.
Since we do not know the owner of the leaf $u$,
we are not able to answer the above question by a simple character comparison.
However, this can be answered again by the aid of the Weiner tree.
Recall that there is an explicit node representing
the reversed string $\rev{x}$ in the Weiner tree and
we know its locus through the updates of the Weiner tree.
Now, the active point can step forward with character $a$
if and only if the node $\rev{x}$ has a (soft or hard) W-link
for character $a$.
Hence, we can answer the above question in $O(\log \sigma)$ time.
In case where we cannot step forward with character $a$,
then we need to create a new edge leading to a new leaf.
Instead of explicitly maintaining the owner of the leaf,
we only maintain the first character $a$ of this edge label.
If the locus of the active point is on an edge,
then we create a new explicit node $u$
representing $x$ in the Ukkonen tree.
Now $u$ has two out-going edges both leading to leaves,
one of which is labeled with $a$ as was described above.
Since $x$ was on an edge,
there was a unique character, say $b$,
such that $b \neq a$ and the W-link of node $\rev{x}$ for character $b$
is defined in the Weiner tree.
Thus the other out-going edge of $u$ is labeled with $b$
in the Ukkonen tree.
Also, by storing the string depth in each active point, 
the whole label of the edge from the parent of $u$ to $u$ can be
easily determined in constant time.
Thus, we are able to eagerly maintain
the whole label of every edge leading to an internal explicit node.

\item The second question is how we can know that
the active point catches up the leaf.
In the preceding discussions, we only proved that
we can find the owner of the leaf \emph{after}
we know that the active point has caught up the leaf.
We observe that the active point catches up the leaf
if and only if the Weiner tree node $v$ representing $a\rev{x}$
is of Weiner degree zero, namely,
the W-link of node $v$ is undefined for any character.
Hence, this question can also be answered by the aid of the Weiner tree
in constant time.
\end{itemize}

The final issue in this forward approach
is how to overcome Difficulty B on the cost for
canonizing active points.
The next lemma implies that
the cost in our fully-online setting
can indeed be amortized by a simple modification
to the original amortization arguments in the semi-online setting.

\begin{lemma} \label{lem:canonization_amortized_cost}
  The total cost for canonizing the active points for all $K$ texts
  in a left-to-right collection $\mathcal{S}$ is $O(N \log \sigma)$.
\end{lemma}

\begin{proof}
Since we swap active points, the owner of each active point
can change during the construction of the Ukkonen tree.
However, our analysis below
does not consider which text is the owner of each active point
and hence it will lead us to simple arguments.

Let $A$ denote any active point and let
$(u_A, c_A, \ell_A)$ denote the reference pair of $A$.
We remark that in our fully-online setting,
this reference pair may not be canonical,
since some other text can split the out-going edge of node $u_A$
whose label begins with $c_A$.
The \emph{potential} of the active point $A$ is $\ell_A$
of the string that hangs off from the explicit node $u_A$.

Suppose we have constructed $\STree(\mathcal{S}_{U[1..i-1]})$,
and that we are given the $i$th update operator $U[i] = (k, a)$
which appends new character $a$ to the $k$th text $\$_k S_k$.
Also, suppose that $A$ is the active point for $\$_k S_k$ at this stage.
Now the algorithm finds the new locus for the active point
$A$ for the updated text $\$_k S_k a$,
while possibly swapping several active points and inserting new leaves.
In this event the algorithm traverses a chain of (virtual) suffix links.
When a canonization is conducted after tracing a virtual suffix link,
then the potential $\ell_A$ decreases at least one.
Also, when the new locus of the active point $A$ is found on
the updated suffix tree $\STree(\mathcal{S}_{U[1..i]})$,
then the potential increases exactly by one with the new character $a$.
Hence, the total number of canonizations
performed for all $N$ added characters is at most $N$.

Each canonization operation requires $O(\log \sigma)$ time
to find the out-going edge whose label begins
with the corresponding character.
Hence, the total cost for canonizations for all $N$ characters is
$O(N \log \sigma)$.
%
%
\end{proof}

Putting the above arguments all together,
we have proven the following theorem.

\begin{theorem} \label{theo:fully-online_ukkonen_forward}
Given a fully-online sequence $U$ of $N$ update operators
for a collection of $K$ left-to-right texts $\mathcal{S}$,
our forward version of Ukkonen's algorithm can update the suffix tree
in a total of $O(N \log \sigma)$ time and $O(N)$ space.
\end{theorem}

A snapshot of left-to-right fully-online suffix tree construction is shown in Fig.~\ref{fig:suffix-tree-fully-online-snapshot},
where the $\$_i$ symbols are omitted for simplicity.

\begin{figure}[p]
  \centerline{
    \includegraphics[scale=0.72]{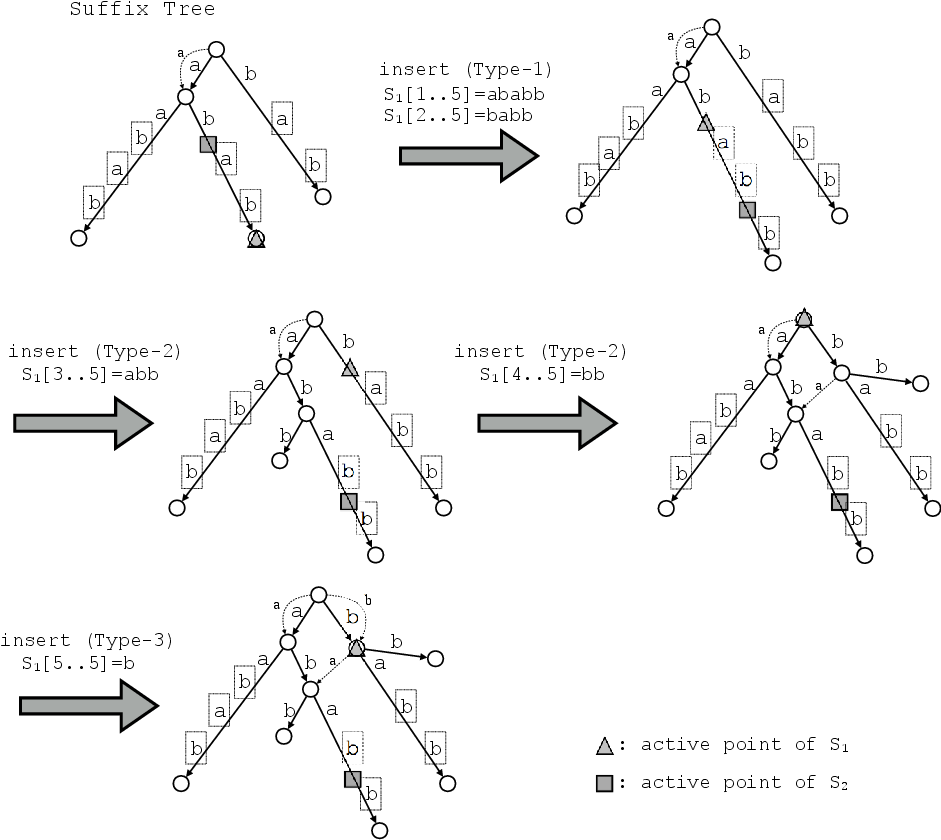}
  }
  \caption{A snapshot of left-to-right fully-online suffix tree construction in the forward approach,
  where we update $\STree(\mathcal{S})$ to $\STree(\mathcal{S'})$
  with $\mathcal{S} = \{S_1 = \mathtt{abab}, S_2 = \mathtt{aabab}\}$
  and $\mathcal{S'} = \{S_1\mathtt{b}, S_2\}$
  (here the terminate symbols $\$_1$ and $\$_2$ are omitted for simplicity).
  Recall that we employ lazy maintenance of leaf ownership,  
  and hence each character within a box is only imaginary and is not computed
  during the updates.
  Due to lazy representation of leaves, 
  we do nothing to insert the Type-1 suffixes of $S_1 \mathtt{b}$.
  The active point of $S_1$ was on a leaf whose owner was $S_2$,
  and then it has extended the leaf.
  Hence, we swap the active points of $S_1$ and $S_2$.
  To start inserting the Type-2 suffixes in decreasing order of length,
  we first insert the longest Type-2 suffix $\mathtt{abb}$ at the locus of the active point of $S_1$.
  With the aid of the Weiner tree, we determine whether the active point can step forward along this edge by character $\mathtt{b}$.
  In this case, the active point cannot step forward, 
  and hence create a new internal node in the middle of this edge.
  After creating a new leaf from the new internal node 
  and its in-coming edge with the first character label $\mathtt{b}$, 
  we determine the label of the in-coming edge of the new internal node
  using Weiner tree.
  Then the active point traces the virtual suffix link from the new internal node $\mathtt{ab}$ to node $\mathtt{b}$.
  This virtual link can be computed by using the suffix link of node $\mathtt{a}$. 
  The next Type-2 suffix is $\mathtt{bb}$, 
  and the active point cannot step forward with $\mathtt{b}$.
  Therefore we create a new internal node in the middle of this edge.
  After creating a new leaf from the new internal node
  and its in-coming edge with the first character label $\mathtt{b}$, 
  we determine the label of the in-coming edge of the new internal node
  using Weiner tree.
  The reversed suffix link is set from this new internal node $\mathtt{b}$ to node $\mathtt{ab}$.
  Then the active point traces the virtual suffix link from the new internal node $\mathtt{b}$ to the root. 
  The next shorter suffix $\mathtt{b}$ is Type-3, since we can step forward with character $\mathtt{b}$ from the root.
  Therefore, we move the active point  from the root to node $\mathtt{b}$ that represents the longest repeating suffix of $S_1 \mathtt{b}$,
  and the reversed suffix link is set from root to the node $\mathtt{b}$.
  Since we have inserted all the Type-2 suffixes,
  the update finishes. 
  }
  \label{fig:suffix-tree-fully-online-snapshot}
\end{figure}

\subsubsection{Backward approach}

In this subsection,
we propose the backward approach
which traces a chain of (virtual) suffix links
in the reversed order and 
inserts new leaves in increasing order of their string lengths.

Suppose we have constructed $\STree(\mathcal{S}_{U[1..i-1]})$
and we are now given an update operator $U[i] = (k, a)$.
Consider the locus of the insertion point of
the shortest Type-2 suffix of the updated text $\S_k S_k a$
in the Ukkonen tree $\STree(\mathcal{S}_{U[1..i-1]})$.
This locus corresponds to the suffix of $\$_k S_ka$
that is exactly one character longer than
the longest Type-3 suffix $\longsuf_{S_{[U1..i-1]}}(\$_k S_ka)$
of $\$_k S_k a$ in the text collection $\mathcal{S}_{U[1..i-1]}$ before update.
In the backward approach
we first find this locus,
and insert the Type-2 suffixes of the updated text $\$_k S_k a$
in increasing order of lengths.
Since we trace the chain of suffix links backward,
we use the reversed suffix links with character labels.
In other words, we maintain \emph{the hard W-links on the Ukkonen tree}.

We also remark that
we do not need to swap active points in this backward approach,
since we begin with the \emph{shortest} Type-2 suffix.
This somewhat simplifies the concept of the algorithm and
might be an advantage over the forward counterpart presented in Section~\ref{sec:forward}.

To find the canonical reference to the locus of the
insertion point of the shortest Type-2 suffix of $\$_k S_k a$,
we use the spanning tree of $\DAWG(\mathcal{T}_{U[1..i]})$
which consists only of the primary edges.
This tree consists of the longest paths from the source
of the DAWG to its nodes,
and hence, it coincides with
the tree of \emph{the reversed hard W-links of the Weiner tree}
(this should not be confused with the hard W-links on the Ukkonen tree
for backward suffix link traversals).
For each $1 \leq i \leq N$,
let $\LPT(\mathcal{S}_{U[1..i]})$ denote this tree.
By the property of DAWGs (and hence that of the equivalence relation),
the following fact holds.
\begin{fact} \label{fact:primary_edges}
  For any $2 \leq i \leq N$,
  if an edge $e$ is a primary edge of $\DAWG(\mathcal{T}_{U[1..i-1]})$,
  then $e$ is a primary edge of $\DAWG(\mathcal{T}_{U[1..i]})$.
\end{fact}
We also use the following fact in our algorithm.
\begin{fact} \label{fact:branching}
  For any substring $x$ of texts in a left-to-right text collection
  $\mathcal{S}$, 
node $x$ is branching (explicit) in $\STree(\mathcal{\mathcal{S}})$ 
iff node $[x]_{\mathcal{S}}$ is branching in $\DAWG(\mathcal{S})$.
\end{fact}
Based on Fact~\ref{fact:branching}, for each $1 \leq i \leq N$,
we will maintain the NMA data structure $\LPT(\mathcal{S}_{U[1..i]})$
and mark its nodes iff they correspond to the
branching nodes of $\STree(\mathcal{S}_{U[1..i-1]})$.
Note that, due to Fact~\ref{fact:primary_edges},
no edges of $\LPT(\mathcal{S}_{U[1..i-1]})$ will be deleted
in $\LPT(\mathcal{S}_{U[1..i]})$
and only new edges will be added.
Hence we can use the NMA data structure on top of this tree.

The next lemma shows how we can efficiently find
the new locus of the active point for the updated text $\$_k S_k a$
in the Ukkonen tree.

\begin{lemma} \label{lem:longset_type3_LPT}
We can compute, in amortized $O(\log \sigma)$ time,
the canonical reference to
the locus of 
the active point of $\$_k S_k a$ on the Ukkonen tree,
using a data structure which requires $O(N)$ space.
\end{lemma}

\begin{proof}
Suppose we have constructed the Ukkonen tree $\STree(\mathcal{S}_{U[1..i-1]})$
in tandem with the Weiner tree $\STree(\mathcal{T}_{U[1..i-1]})$
and $\LPT(\mathcal{S}_{U[1..i-1]})$.
A node $v$ of $\LPT(\mathcal{S}_{U[1..i-1]})$ is marked
iff its corresponding node $\rev{v}$ in
the Weiner tree $\STree(\mathcal{T}_{U[1..i-1]})$
has at least two W-links defined, namely,
$\sw_{c}(\rev{v}) = \sw_{c'}(\rev{v}) = 1$
with at least two distinct characters $c \neq c'$.
This in turn implies that the corresponding node of the
(implicitly maintained) DAWG is branching.
Every marked node of $\LPT(\mathcal{S}_{U[1..i-1]})$ is linked 
to its corresponding node of the Ukkonen tree $\STree(\mathcal{S}_{U[1..i-1]})$
which is also branching by Fact~\ref{fact:branching}
(see also Figure~\ref{fig:LPT}).
We also maintain an NMA data structure on $\LPT(\mathcal{S}_{U[1..i-1]})$.

\begin{figure}[tb]
  \centerline{
    \includegraphics[width=1.0\textwidth]{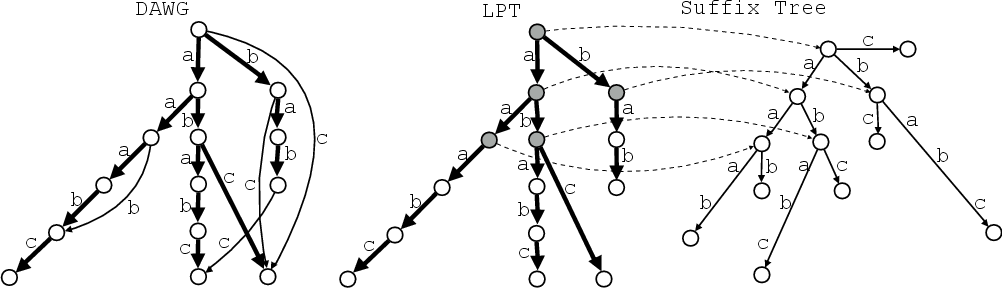}
  }
  \caption{Illustration for $\DAWG(\mathcal{S}_{U[1..14]})$,
  $\LPT(\mathcal{S}_{U[1..14]})$, 
  and the Ukkonen tree $\STree(\mathcal{S}_{U[1..13]})$ before update,
  where
  $\mathcal{S}_{U[1..13]} = \{S_1 = \mathtt{aaab}, S_2 = \mathtt{ababc}, S_3 = \mathtt{bab}\}$
  and $\mathcal{S}_{U[1..14]} = \{S_1\mathtt{c}, S_2, S_3\}$.
  For simplicity, we here omit the terminate symbols $\$_1$, $\$_2$, and $\$_3$.
  The bold solid arrows represent the primary edges of $\DAWG(\mathcal{S}_{U[1..14]})$,
  the gray nodes are the marked nodes of $\LPT(\mathcal{S}_{U[1..14]})$, and
  the dashed arrows represent the links between the marked nodes of 
  $\LPT(\mathcal{S}_{U[1..14]})$ and the corresponding branching nodes of
  $\STree(\mathcal{S}_{U[1..13]})$.
  The longest repeating suffix of $S_1 \mathtt{c}$ in $\mathcal{S}_{U[1..14]}$
  is $\mathtt{abc}$,
  and hence we perform an NMA query from node $\mathtt{abc}$ on $\LPT(\mathcal{S}_{U[1..14]})$, obtaining node $\mathtt{ab}$.
  We then access the suffix tree node $\mathtt{ab}$ using the link from $\LPT(\mathcal{S}_{U[1..14]})$, and obtain the canonical reference $(\mathtt{ab}, \mathtt{c}, 1)$ to $\mathtt{abc}$ on the Ukkonen tree $\STree(\mathcal{S}_{U[1..13]})$ before update.
  }
  \label{fig:LPT}
\end{figure}

Given an update operator $U[i] = (k, a)$,
we first update the Weiner tree to $\STree(\mathcal{T}_{U[1..i]})$.
This introduces at most two new hard W-links,
one for the new leaf and one for its parent.
This means that these edges are also inserted to $\LPT(\mathcal{S}_{U[1..i-1]})$ 
and we then obtain $\LPT(\mathcal{S}_{U[1..i]})$.
Because of these new edges,
at most two DAWG non-branching nodes can become branching.
We mark their corresponding nodes in $\LPT(\mathcal{S}_{U[1..i-1]})$,
and link them to the corresponding Ukkonen tree nodes
\emph{only after}
we have built the updated Ukkonen tree $\STree(\mathcal{S}_{U[1..i-1]})$.
This is because the corresponding nodes of $\STree(\mathcal{S}_{U[1..i-1]})$
before the update are still non-branching (see Fact~\ref{fact:branching}).

Let $\rev{y}$ be the insertion point of the leaf $a T_k \$_k$
in the Weiner tree
which is the longest repeating prefix of $a T_k \$_k$ in
the right-to-left text collection $\mathcal{T}_{U[1..i]}$.
By the definition of $\LPT(\mathcal{S}_{U[1..i]})$,
there is a node in $\LPT(\mathcal{S}_{U[1..i]})$ which represents $y$.
We conduct an NMA query from $y$ on $\LPT(\mathcal{S}_{U[1..i]})$,
and let $v$ be the NMA of $y$.
Let $\ell = |y| - |v|$,
and let $c$ be the label of the first edge in the path
from $v$ to $y$.
We move from $v$ to its corresponding node $x$ in
the Ukkonen tree $\STree(\mathcal{S}_{U[1..i-1]})$.
Then, $(x, c, \ell)$ is a reference to the insertion point
of the shortest Type-2 suffix of $\$_k S_k a$.
Since $v$ is the NMA of $y$ in $\LPT(\mathcal{S}_{U[1..i]})$,
and since updating $\$_k S_k$ to $\$_k S_k a$ does not explicitly insert
any suffix of $\$_k S_k a$ that is shorter than the longest repeating suffix 
of $\$_k S_k a$ in $\mathcal{S}_{U[1..i]}$,
this reference is canonical by Fact~\ref{fact:branching}.

Clearly the total size of the above data structures
is linear in the total length $N$ of the texts in
the final text collection $\mathcal{S}$.
We analyze the time complexity.
We can find the insertion point $y$ of the new leaf
in the Weiner tree in amortized $O(\log \sigma)$ time
due to Theorem~\ref{theo:fully-online_weiner}.
Using the link from the node $y$ in $\LPT(\mathcal{S}_{U[1..i]})$,
the corresponding node in the Ukkonen tree $\STree(\mathcal{S}_{U[1..i-1]})$
can be found in $O(1)$ time.
Updating $\LPT(\mathcal{S}_{U[1..i-1]})$ to
$\LPT(\mathcal{S}_{U[1..i]})$ takes $O(\log \sigma)$ amortized time.
Inserting a new node
and querying an NMA from a given node takes amortized $O(1)$ time.
We can link a new marked node of $\LPT(\mathcal{S}_{U[1..i]})$
to the corresponding new branching node of $\STree(\mathcal{S}_{U{1..i}})$
in $O(1)$ time,
since it is easy to remember this new branching node when updating 
$\STree(\mathcal{S}_{U[1..i-1]})$ to $\STree(\mathcal{S}_{U[1..i]})$.
Hence, the total amortized bound is $O(\log \sigma)$.
\end{proof}

Let $w$ and $w'$ denote the strings
that are represented by the loci of the insertion points
of the shortest and longest new leaves w.r.t. the update operator
$U[i] = (k, a)$.
Let $q = |w'|-|w|+1$ be the number of new leaves to be inserted
in the Ukkonen tree.
Our backward approach terminates the $i$th update
after inserting the $q$th new leaf.
How do we compute this value $q$?
If $(x, c, \ell)$ is the canonical reference to the locus for $w$,
then $|w| = |x|+\ell$,
and hence what remains is how to compute $|w'|$.
We note that $w'$ is the longest suffix of $\$_k S_k$
which has at least one more occurrence in $\mathcal{S}_{U[1..i]}$
immediately followed by another character $b \neq a$.
This is because any longer suffix of $\$_k S_k$ is immediately
followed only by $a$, and will thus correspond to existing leaves
in the updated Ukkonen tree.
These two occurrences of $w'$ must be immediately preceded by
distinct characters, say $c$ and $d$, in the left-to-right text collection
$\mathcal{S}_{U[1..i]}$ since otherwise there will be a longer suffix of
$\$_k S_k$ which has at least one more occurrence in $\mathcal{S}_{U[1..i]}$,
a contradiction.
Also, $\rev{w'}c$ and $\rev{w'}d$ occur in the right-to-left
text collection $\mathcal{T}_{U[1..i-1]}$ before the $i$th update.
Thus, $\rev{w'}$ is represented by an explicit node in the Weiner tree
$\STree(\mathcal{T}_{U[1..i-1]})$.
Since this node is on the path from the leaf for $T_k \$_k$ to
the root of the Weiner tree,
and since it is the deepest node with the hard W-link for character $a$,
we visit this node during the update of the Weiner tree.
Hence, we can compute $|w'|$ in $O(\log \sigma)$ amortized time
by the aid of the Weiner tree.


The cost to trace the suffix link chains in this backward approach
is exactly the same as that in the forward approach.
Hence, the total cost is for suffix link chain traversals is $O(N \log \sigma)$
for all $1 \leq i \leq N$ by Lemma~\ref{lem:canonization_amortized_cost}\footnote{In the preliminary versions~\cite{TakagiIA15,TakagiIA16} of this paper, a simplified version of the \emph{suffix tree oracle}~\cite{FischerG2015} was used to obtain the same bound. However, we do not need it any more due to our amortization argument of Lemma~\ref{lem:canonization_amortized_cost}.}.

The lower bound of Lemma~\ref{lem:leaf_ownership_lowerbound}
also applies to this backward approach.
Hence, we do not maintain the leaf ownerships,
and we label the edges leading to the leaves
only with their first characters.

We have shown the following:
\begin{theorem} \label{theo:fully-online_ukkonen_backward}
Given a fully-online sequence $U$ of $N$ update operators
for a collection of $K$ left-to-right texts $\mathcal{S}$,
our backward version of Ukkonen's algorithm can update the suffix tree
in a total of $O(N \log \sigma)$ time and $O(N)$ space.
\end{theorem}

\section{Conclusions and open problems}

In this paper, we considered construction of
the suffix tree and the DAWG of the fully-online multiple texts,
where new characters can be added to any of the texts.

Our contribution is two-folds.

First, we proposed the fully-online version
of Weiner's suffix tree construction algorithm
for a collection of $K$ right-to-left growing texts.
This algorithm runs in $O(N \log \sigma)$ time with $O(N)$ space,
where $N$ is the total length of the texts in the collection
and $\sigma$ is the alphabet size.
We showed that the direct application of Weiner's algorithm
to our fully-online setting takes $\Theta(N \min(K, \sqrt{N}))$ time
(Lemma~\ref{lem:Weiner_lower_bound}),
and showed that how it can be modified to run in $O(N \log \sigma)$ time
with a novel use of NMA data structures which occupy $O(N)$ total space
(Theorem~\ref{theo:fully-online_weiner}).
We also showed an algorithm which simulates soft W-links with hard W-links
and these NMA data structures in $O(\log \sigma)$ time per query,
which immediately gives us an $O(N \log \sigma)$-time construction algorithm
for an $O(N)$-space representation of the DAWG
for a fully-online left-to-right text collection (Corollary~\ref{coro:fully-online_dawg}).

Second, we proposed two variants of
the fully-online version of Ukkonen's suffix tree
construction algorithm for a collection of $K$ left-to-right growing texts.
We showed that explicit maintenance of the owners of leaves
requires us super-linear cost (Lemma~\ref{lem:leaf_ownership_lowerbound})
in the worst case.
Then, we proposed the first variant called the forward approach,
which runs in $O(N \log \sigma)$ time with $O(N)$ space.
The key to this forward approach is the notion of swapping active points
and the efficient algorithm for answering leaf ownerships
in a spacial case which happens during the construction of the suffix tree.
The second variant called the backward approach traces
a virtual suffix link chain in the reversed direction to the forward approach,
and also works in $O(N \log \sigma)$ time with $O(N)$ space.


There are many intriguing open problems for the left-to-right
fully-online suffix tree construction.
Examples are the following:
\begin{enumerate}
\item[(1)] Is it possible to maintain the Ukkonen's tree
for a left-to-right text collection \emph{without} the aid of
the Weiner tree for the corresponding right-to-left text collection?

\item[(2)] Is there a data structure which maintains leaf ownerships
in an \emph{implicit} manner,
so that the ownership of an \emph{arbitrary} leaf can be efficiently
answered upon query, at any stage of the construction algorithm?

\item[(3)] Our bound is amortized, namely, for each new character
    our algorithm takes $O(\log \sigma)$ amortized time.
    Is it possible to de-amortize it,
    e.g. by using techniques in~\cite{BreslauerI13,abs-1302-3347,FischerG2015}?

\item[(4)] Is it possible to extend our approach for
  a \emph{bidirectional} fully-online text collection,
  where each text can grow both directions?
  There is a $O(n \log \sigma)$-time $O(n)$-space algorithm
  for constructing the suffix tree for a single bidirectional text
  of length $n$~\cite{Inenaga03}.
  We note that for a bidirectional fully-online text collection,
  we cannot use terminal $\$_k$ symbols either ends of each text
  during the updates.
\end{enumerate}



\end{document}